\documentclass[a4paper,11pt]{article}
\usepackage{PRIMEarxiv}
\usepackage[ruled]{algorithm2e} % For algorithms

\usepackage{amsfonts}       % Blackboard math symbols
\usepackage{geometry}
\usepackage{amsmath}
\usepackage{amssymb}
\usepackage{amsthm}
\usepackage{authblk}
\usepackage{booktabs}  
\usepackage{hyperref}% Professional-quality tables
\usepackage{boxedminipage}
\usepackage{caption} 
\usepackage{cleveref}
\usepackage{comment}
\usepackage{enumitem}
\usepackage{fancyhdr}       % Header
\usepackage{graphicx}       % Graphics
\usepackage[utf8]{inputenc} % Allow UTF-8 input
\usepackage{lipsum}
\usepackage{microtype}      % Microtypography
\usepackage{multirow}
\usepackage{natbib}  
\usepackage{nicefrac}       % Compact symbols for 1/2, etc.
\usepackage{pifont}
\usepackage{pgfplots}
\usepackage[T1]{fontenc}    % Use 8-bit T1 fonts
\usepackage{thm-restate}
\usepackage{threeparttable}
\usepackage{tikz}
\usepackage{url}            % Simple URL typesetting
\usepackage{xcolor}         % Colors
\usepackage{xr}
\newcommand{\diverse}{maximally diverse}
\newcommand{\egalopt}{egalitarian-optimal}
\newcommand{\utilopt}{utilitarian-optimal}

\setlist{nolistsep}
\usetikzlibrary{arrows}
\usetikzlibrary{patterns}
\SetKw{Break}{break}

\newtheorem{lemma}{Lemma}%
\newtheorem{definition}{Definition}%
\newtheorem{theorem}{Theorem}%
\newtheorem{corollary}{Corollary}%

\newtheorem{claim}{Claim}

\newcommand{\rurals}{\mathcal{S}}

\allowdisplaybreaks

\newcommand{\bottom}{\text{bottom}}

\newcommand{\numM}{\eta}

\newcommand{\stabset}{\mathcal{M}}

\newcommand{\w}{\omega}

\newcommand{\divv}{d^{\mathrm{div}}}

\newcommand{\rural}{\mathcal{R}}

\title{\bf Maximally Diverse Stable Matchings: Optimizing Arbitrary Institutional Objectives}

\newtheorem{remark}{Remark}

\newcommand{\Fam}{\mathcal{F}}

\author[1,2]{Gergely Csáji}
\author[3,4]{Zhaohong Sun}

\affil[1]{Centre for Economic and Regional Studies, Eötvös Loránd University, Hungary}
\affil[2]{Budapest University of Technology and Economics, Hungary}
\affil[3]{Kyushu University, Japan}
\affil[4]{AI Lab, CyberAgent, Japan}

\date{\vspace{-10mm}}

\begin{document}
\maketitle

\begin{abstract}
Stable matching provides the theoretical foundation for many centralized market mechanisms. However, modern matching markets often require additional distributional objectives, such as diversity constraints and regional quotas, which may conflict with classical stability. Existing approaches typically resolve this tension by weakening stability in order to satisfy these distributional objectives.

In this paper, we take the opposite perspective: rather than weakening stability, we optimize institutional objectives over the set of classical stable matchings. We first study matching with diversity objectives and propose a specialized \emph{Diversity-Constrained Deferred Acceptance} (DC-DA) algorithm that efficiently computes stable matchings minimizing diversity violations under the challenging one-to-all counting convention. More generally, we introduce a polynomial-time framework for optimizing arbitrary institution-specific objectives defined over assigned student sets. Our key technical contribution is a general reduction that exploits the structure of stable matchings to transform these objectives into tractable optimization problems. The framework applies to a broad range of distributional goals, including diversity objectives under different counting conventions and other composition-based objectives. Experiments on synthetic matching markets confirm that our algorithms achieve strong distributional performance while preserving classical stability.

Our results provide a new perspective on the relationship between stability and distributional objectives. Rather than sacrificing stability to achieve institutional goals, we characterize the best distributional outcomes attainable within the stable matching space. This allows us to quantify the trade-off between stability and distributional objectives, measuring the loss incurred by requiring stability.
\end{abstract}

\section{Introduction}

Matching markets play a central role in modern market design, with applications ranging from school choice and college admissions to medical residency matching. The theory of stable matchings, pioneered by \citet{GaSh62a}, provides the foundation for these markets, and the Deferred Acceptance (DA) algorithm has become the standard mechanism for assigning students to schools \citep{AbSo03b,APR05a} and medical graduates to residency programs \citep{Roth84a,RoPe99a}. In these settings, stability, defined by the absence of blocking pairs, is widely regarded as a fundamental design principle. It eliminates justified envy and is essential for the fairness, credibility, and long-term sustainability of centralized matching mechanisms \citep{Roth08a}.

However, modern matching markets increasingly face institutional objectives that extend beyond capacity alone. Policymakers and institutions often seek matchings that satisfy various distributional goals reflecting fairness, diversity, or operational requirements. Representative examples include matching with diversity constraints \citep{AbSo03b} and matching with inseparable members \citep{McMa10a}. For instance, school districts and universities may seek to admit balanced cohorts by imposing lower and upper quotas on students from different socioeconomic, geographic, or demographic groups \citep{EHYY14a}. In daycare matching, parents strongly prefer siblings to be assigned to the same facility \citep{STM+23a}. Similarly, refugee resettlement programs aim to allocate refugee families across localities while keeping family members together and respecting local capacity and resource constraints \citep{ACGS18a}. These examples illustrate that institutions frequently pursue complex objectives defined over the composition of their assigned students, rather than merely over individual assignments.

Reconciling stability with distributional objectives has proved to be one of the central challenges in matching theory. In particular, rigidly enforcing hard distributional constraints may result in the nonexistence of stable matchings \citep{KaKo17b,DKT23a}. To overcome this difficulty, much of the literature has focused on modifying institutional choice functions so that soft distributional bounds can be incorporated into the matching process. Consequently, many of these approaches replace classical stability with weaker stability notions tailored to the modified choice rules \citep{HHKS+17a,AzSu21a,KaKo23a}.

In contrast, we ask the reverse question: \textbf{if stability is treated as a non-negotiable requirement, how far can institutional objectives be optimized?} Rather than weakening stability to accommodate distributional goals, we optimize distributional objectives over the set of classical stable matchings. This perspective naturally leads to the following algorithmic question: given an arbitrary objective defined on the composition of students assigned to each institution, can we efficiently find a stable matching that minimizes the resulting global objective?

We begin with matching under diversity objectives, one of the most widely studied forms of distributional constraints. A fundamental distinction in this literature lies in how students belonging to multiple categories are counted. Under the one-to-all convention, a student contributes simultaneously to every category to which they belong, as in gap-year programs in Israel \citep{GNKR19a}. Under the one-to-one convention, each student is counted toward exactly one eligible category, as in affirmative-action systems for college admissions in India \citep{SoYe22a}, school choice in Chile \citep{CEE+22a}, and vaccine allocation \citep{PSUY21a,AzBr21a}.

Existing work has predominantly focused on the one-to-one convention because it admits elegant reserve-based formulations and efficient algorithms. This line of work includes reserve-seat mechanisms \citep{AyTu17a,KHIY17a,BCC+19a}, smart reserves \citep{SoYe22a}, and multiple ranks of reserves \citep{AzSu21a}. In contrast, the one-to-all convention is computationally much more challenging, since a student simultaneously contributes to multiple categories. Whether one-to-all diversity objectives can be optimized efficiently while preserving stability has remained largely open.
We answer this question positively. Moreover, our approach is substantially more general: the proposed framework is not restricted to diversity objectives but applies to arbitrary institutional objectives defined over the set of students assigned to each institution.

\subsection{Our contributions}

Our contributions are summarized as follows.

First, we introduce a general framework for optimizing arbitrary institutional objectives over the set of stable matchings. Rather than restricting attention to diversity constraints, we allow each institution to specify an arbitrary polynomial-time computable objective defined over its assigned students. Our key technical contribution is an exact reduction showing that such set-valued objectives can be represented by linear edge weights on the domain of stable matchings. This reduction transforms the original optimization problem into a classical minimum-weight stable matching problem, yielding polynomial-time algorithms for a broad class of institutional objectives while preserving classical stability.

Second, we apply this framework to matching with diversity objectives, one of the most important classes of distributional constraints. We focus in particular on the largely overlooked one-to-all counting convention, where each student simultaneously contributes to every category to which they belong. While existing approaches predominantly study the one-to-one convention due to computational and structural advantages, we show that one-to-all diversity objectives can also be optimized efficiently among stable matchings. Moreover, we develop a specialized \emph{Diversity-Constrained Deferred Acceptance} (DC-DA) algorithm that computes diversity-optimal stable matchings without requiring the general weight transformation or minimum-weight stable matching computation.

Third, we establish structural properties of the resulting optimization problems. We show that optimal stable matchings under general institutional objectives form a sublattice of the stable matching lattice, and therefore admit both student-optimal and institution-optimal solutions. These results extend classical lattice properties of stable matchings to settings with optimized distributional objectives and provide a principled way to quantify the trade-off between stability and institutional goals.

Finally, we evaluate our algorithms through computational experiments. The results demonstrate that our methods achieve strong distributional performance while always preserving classical stability, producing matchings with zero blocking pairs. In contrast, approaches that prioritize distributional objectives by modifying stability may achieve better distributional outcomes at the cost of stability violations. Our experiments illustrate the practical trade-off between distributional objectives and classical stability.

\begin{table*}[t]
\centering
\small
\caption{Comparison of state-of-the-art diversity matching algorithms and their theoretical properties. Our approach gives the first efficient algorithm that works for one-to-all counting, guarantees some form of stability and can handle overlapping student types. It is also the first approach to guarantee standard stability while optimizing diversity goals.}
\label{tab:comparison}
\begin{tabular}{lccccc}
\toprule
\textbf{Algorithm} 
& \textbf{Standard} 
& \textbf{Types per} 
& \textbf{Lower \& Upper} 
& \textbf{Counting} 
& \textbf{Arbitrary} \\
& \textbf{Stability} 
& \textbf{student} 
& \textbf{bounds} 
& \textbf{Convention} 
& \textbf{Objectives} \\
\midrule

Standard DA 
& $\checkmark$ & -- & & N/A & \\

Smart-Hard-Reserves (\cite{SoYe22a}) 
& & Any & lower or upper 
& One-to-One & \\

Multi-Rank Smart Reserves (\cite{AzSu21a}) 
& & Any & $\checkmark$ & One-to-One & \\

DA-OT* (\cite{KHIY15a}) 
& & Any & $\checkmark$ & One-to-One & \\

Multi-Priority Soft Reserves (\cite{EHYY14a}) 
& & 1 & $\checkmark$ & N/A & \\

\cite{Huan10a} &  & Laminar & $\checkmark$ & One-to-All & \\

\midrule
\textbf{DC-DA (Ours)} 
& $\checkmark$ & Any & $\checkmark$ & One-to-All & \\

\textbf{General Algorithm (Ours)} 
& $\checkmark$ & Any & $\checkmark$ & Both & $\checkmark$ \\

\bottomrule
\end{tabular}
\end{table*}

% \subsection{Paper Structure}
% The remainder of the paper is organized as follows. In Section~\ref{sec:relwork}, we review related work. Section~\ref{sec:Preliminaries} outlines the preliminaries and the basic model. In Section~\ref{sec:DA-maxviol} , we provide specialized DA-based algorithms for minimizing diversity violations with single and overlapping categories. Section~\ref{sec:main} introduces our main theoretical result: the general algorithm for minimizing arbitrary set functions over stable matchings.

\subsection{Paper Structure}

The remainder of the paper is organized as follows. 
In Section~\ref{sec:relwork}, we review related work. 
Section~\ref{sec:Preliminaries} introduces the preliminaries, the basic matching model, and the formulation of distributional objectives. Section~\ref{sec:DA-maxviol} focuses on diversity objectives as an important special case and develops specialized Deferred Acceptance-based algorithms for minimizing diversity violations under different category structures. Section~\ref{sec:main} then presents our main theoretical result: a general polynomial-time framework for optimizing arbitrary institutional objectives over stable matchings. Section~\ref{sec:simulations} presents computational experiments evaluating the performance of our algorithms and illustrating the trade-off between distributional objectives and classical stability. 
% Finally, Section~\ref{sec:conclusion} concludes the paper.

\section{Related Work}\label{sec:relwork}

A central challenge in school choice with diversity constraints is how to incorporate distributional requirements while preserving desirable properties such as fairness and efficiency. Early work studied this problem in environments where each student has a single type. \citet{AbSo03b} incorporated type-specific upper quotas into deferred acceptance (DA) and top trading cycles (TTC), achieving within-type fairness and Pareto efficiency among feasible matchings, while \citet{Abdu05a} established strategy-proofness of DA under responsive preferences over quota-satisfying sets. However, imposing hard quotas can lead to unintended welfare consequences \citep{Koji12a} and, more broadly, may render fairness and non-wastefulness incompatible \citep{EHYY14a}. This has led to alternative approaches based on minority reserves \citep{HYY13a} and soft bounds with modified fairness notions \citep{EHYY14a}, as well as a characterization of admissible choice functions \citep{EcYe15a}.

A second line of work highlighted an additional difficulty: in many applications, students belong to multiple (overlapping) types. In such settings, two conventions arise. Under the one-for-all convention, a student consumes the reserved seats of all her types, whereas under one-for-one, she occupies only one \citep{SoYe22a}. The one-for-all convention is computationally challenging and often incompatible with fairness or incentive properties \citep{BFIM10a,GNKR19a,AGS20a}, while the one-for-one convention is more widely adopted. Existing approaches under one-for-one typically assign types in a greedy or sequential manner—via strict preferences, tie-breaking, or dynamic priorities \citep{AyTu16a,KHIY17a,BCC+19a,CEE+19a,EHYY14a}—which may fail to fully achieve diversity goals when types overlap.
To address this limitation, \citet{SoYe20a} proposed the smart reserves approach, which allows flexible assignment of types and improves the utilization of reserved seats. However, their framework is restricted to a single reserve tier and cannot capture richer constraints. Subsequent work extended this approach to multiple reserve tiers, allowing for both lower and upper quotas as well as hierarchical priorities across types \citep{AzSu21a,AzSu25a}.
Under the one-for-all convention, alternative approaches rely on heuristic or restricted mechanisms, often at the cost of fairness, strategy-proofness, or general applicability \citep{GNKR19a,AGS20a,BCC+19a,Aziz19b}. Related papers consider ratio constraints \citep{NgVo19a}, sequential reserve processing \citep{DPS20a}, and vertical and horizontal reservation systems \citep{SoYe22a}, each addressing specific institutional features but not the general problem of optimally allocating overlapping types.

Closely related is the classified stable matching problem introduced by \citet{Huan10a}, where each institution can classify the set of students independently, and enforce upper and lower quotas for each category. He showed that deciding the existence of a classified stable matching is NP-hard, but becomes efficiently solvable when the categories form nested systems. This was extended later by~\citet{fleiner2016matroid} and~\citet{yokoi2017generalized} to matroidal settings.  
However, the algorithms by \citet{Huan10a}, \citet{fleiner2016matroid}, and \citet{yokoi2017generalized} define stability by requiring that blocking coalitions themselves satisfy the constraints, which effectively restricts justified envy only within categories. In our approach, stability is defined in the original sense, without taking into account type or category-specific quotas.

Another critical dimension of modern matching markets is the accommodation of joint preferences, most notably the assignment of couples in labor markets and siblings in school choice and daycare. Foundational work by \citet{roth1984evolution} highlighted that when participants express joint preferences, stable matchings may fail to exist. Furthermore, \citet{ronn1990np} established that the problem of determining whether a stable matching exists in the presence of couples is NP-complete. This hardness persists under additional restrictions, such as couples only ranking pairs within the same hospital \cite{McMa10a}, or when all hospitals share a common ranking of doctors \cite{BIS11a}. In random markets, \cite{KPR13a} proved that when the market grows large and the number of couples grows slowly (e.g., $O(\sqrt{n})$), a stable matching exists with probability approaching one. On the algorithmic side, the NRMP uses a heuristic based on the incremental algorithm of \cite{roth1999redesign}. Alternative approaches include the Scarf algorithm for finding fractional matchings \cite{BFI16a}, integer and constraint programming methods \cite{MOPU07a,BMM14a,MMT17a}.

Similarly to resident allocation with couples, in school choice and daycare matching, families face substantial logistical costs when siblings are separated, necessitating priority structures that strongly favor joint assignments. \citet{sun2025probabilistic} demonstrate that while sibling complementarities technically jeopardize theoretical stability much like couples, stable matchings can be found with high probability in large random markets where institutional priorities are highly correlated.

Optimization over stable matchings has been widely considered in previous literature, but mostly for problems that can be easily represented by linear edge weights $\w:E\to \mathbb{R}$ over any set of edges. The first algorithm to compute a minimum $\w$-weight stable matching is due to~\citet{irving1987efficient}. They proved that stable matchings can be compactly represented by a rotation poset and used this result to show that computing a minimum weight stable matching reduces to finding a minimum weight cut in a digraph. This algorithm was extended to the many-to-many case by~\citet{eirinakis2012finding}.
A different approach to compute optimal stable matchings is to describe the convex hull of stable matchings as a Linear Program (LP). \citet{vande1989linear} provided the first description for the stable matching polytope, which was later simplified and extended by~\citet{rothblum1992characterization} and~\citet{roth1993stable}. \citet{baiou2000stable} provided a polyhedral characterization and algorithms for maximum-weight stable many-to-many matchings. 
 Regarding fairness objectives, \citet{irving1987efficient} gave an algorithm to find an egalitarian stable matching (minimizing the total rank of all agents on both sides), \citet{gusfield1987three} provided one for the minimum regret stable matching (minimizing the rank of the worst agent on either side). \citet{kato1993complexity} proved that finding a sex-equal stable matching (that minimizes the differences in total ranks of the two sides) however, is NP-hard. So is finding a median stable matching, where each agent receives its median partner (counted with multiplicity) among all stable matchings as shown by \cite{cheng2010understanding}.

Recently, \citet{faenza2025minimum} developed a complex framework for optimizing  objectives over stable matchings by representing them as minimum cuts in a digraph. However, their framework does not give efficient algorithms generally, as their approach to deciding minimum cut representability requires checking whether certain (global) complex conditions hold for all (exponentially many) stable matchings. They provide efficient algorithms for some simpler cases, where this can be done efficiently, such as when the goal is to decide if  pairs of siblings can be matched together in a stable matching. In contrast, our approach is much simpler and is efficient for basically any reasonable institutional objective functions and thus can extend the algorithm of \citet{faenza2025minimum} to arbitrary-sized families of siblings, as well as many other problems.

Finally, our work is related to broader models of matching with distributional constraints, including regional quotas \citep{KaKo15a,AGSW19a} and housing allocation with ethnic caps \citep{BCH+18a}, as well as a growing literature in matching and AI on constrained allocation \mbox{\citep{GIKY+16a,STY21a,ZYBY18a,ABY22a}}. These works consider different constraint structures or objectives and do not address the combination of overlapping types, flexible diversity goals, and optimal reserve allocation that we study.

%%%%%%%%%%%%%%%%%%%%%%%%%
% Section: Preliminaries
%%%%%%%%%%%%%%%%%%%%%%%%%
\section{Preliminaries}
\label{sec:Preliminaries}

In this section, we formally define the matching models, the various constraints, and the objective functions studied in this paper. We begin with the classic many-to-one stable matching model and then extend it to incorporate diversity categories and general institutional objectives.

\subsection{The Basic Matching Model and Stability}

A basic matching market consists of a finite set of students $S$ and a finite set of institutions $I$. Each institution $i \in I$ has a total capacity of $q_i$ positions. The set of acceptability relations is modeled by a bipartite graph $(S \cup I, E)$, where an edge $e = (s, i) \in E$ indicates that student $s$ and institution $i$ find each other acceptable. The set of edges incident to an agent $v \in S \cup I$ is denoted by $E(v)$. We assume without loss of generality that all $v\in S\cup I$ are incident to at least one edge.
Each student $s \in S$ has a strict preference ordering $\succ_s$ over her acceptable institutions (i.e.,
$I\cap E(s)$). Similarly, each institution $i \in I$ has a strict priority ordering $\succ_i$ over its acceptable students (i.e., $S\cap E(i)$). We assume that the option of remaining unmatched, denoted by $\emptyset$, is strictly less preferred than any acceptable partner.
For each agent $a\in S\cup I$, let
$M(a)=\{\,b\mid (a,b)\in M\,\}$ denote the set of agents matched to $a$ under a matching $M$.
A \emph{matching} $M\subseteq E$ is a subset of edges such that for every student $s\in S$, we have $|M(s)|\le 1$, and for every institution $i\in I$, we have $|M(i)|\le q_i$.
%

% \begin{definition}[Matchings and Stability]
    % \begin{itemize}
    %     \item A \emph{matching} $M \subseteq E$ is a subset of edges such that for every student $s \in S$, we have $|M(s)| \le 1$, and for every institution $i \in I$, we have $|M(i)| \le q_i$.
    %     \item Given a matching $M$, an edge $(s, i) \notin M$ \emph{blocks} $M$ if $i \succ_s M(s)$ and either $|M(i)| < q_i$ or there exists some student $s' \in M(i)$ such that $s \succ_i s'$. 
    %     \item A matching $M$ is \emph{stable} if it admits no blocking edges. We denote the set of all stable matchings by $\stabset$.
    % \end{itemize}
% \end{definition}

Stability is a classical solution concept \citep{GaSh62a},  requiring that no student prefers another institution that either has a vacant seat or would prefer the student to one of its currently assigned students.

\begin{definition}[Stable Matching]
\label{def:stability}
Given a matching $M$, an edge $(s, i) \notin M$ \emph{blocks} $M$ if $i \succ_s M(s)$ and either $|M(i)| < q_i$ or there exists some student $s' \in M(i)$ such that $s \succ_i s'$.
A matching $M$ is \emph{stable} if it admits no blocking edges. We denote the set of all stable matchings by $\stabset$.
\end{definition}

\paragraph{Deferred Acceptance (DA) and the \texttt{BreakMatch} operation.}
We use the student-proposing Deferred Acceptance (DA) algorithm to compute stable matchings. Under DA, students iteratively apply to institutions in the order of their preference lists. Each institution tentatively holds its most-preferred applicants up to its capacity and rejects the remaining applicants. The process terminates when no further proposals are made.
\\
For an edge $(s,i)$, the operation \texttt{BreakMatch} is defined as follows. If $s$ is not currently matched to $i$, the operation has no effect. Otherwise, the pair $(s,i)$ is removed from the current matching, and student $s$ resumes the DA process from the next institution on her preference list. In addition, institution $i$ becomes permanently unavailable to $s$ and to every student $s’$ such that $s \succ_i s’$.
That is, any future proposal from such a student to $i$ is ignored. This restriction maintains the invariant that student $s$ cannot acquire justified envy toward any student assigned to $i$ in subsequent iterations. Furthermore, the least-preferred student assigned to each institution weakly improves with respect to the institution’s preference ordering throughout the process.

The key result of this paper depends on the possible set of edges for each institution $i\in I$ that can arise in a stable matching.

\begin{definition}[Stable Sets]
\label{def:stable_sets}
Let $\mathcal{E}(i)=\{$
$E_i^1, \dots, E_i^{k_i}\}$ denote the distinct subsets of $E(i)$ such that there exists a stable matching $M$ where $M(i) = E_i^j$. We refer to these as the \emph{stable sets} of institution $i$.
\end{definition}

\subsection{Diversity Constraints and General Institutional Objectives}

A central and well-studied form of distributional constraints is matching with diversity goals. In this setting, we introduce a finite set of categories $C$, and each student $s \in S$ is associated with a subset $C(s) \subseteq C$. For each category $c \in C$, institution $i$ specifies a soft lower bound $\ell_i^c$ and a soft upper bound $u_i^c$. Let $\numM_i^c(M)$ denote the number of students of category $c$ assigned to $i$ under matching $M$. The objective is to find a stable matching that minimizes violations of these bounds, which can be formalized as follows.

\begin{definition}[Diversity Violations and Optimality]
Given a stable matching $M$, define the \emph{total diversity violation} as
\[
\divv_1(M)=\sum_{i\in I}\sum_{c\in C}\big((\ell_i^c-\numM_i^c(M))^+ + (\numM_i^c(M)-u_i^c)^+\big),
\]
and the \emph{maximum diversity violation} as
\[
\divv_\infty(M)=\max_{i\in I,\,c\in C} \big\{(\ell_i^c-\numM_i^c(M))^+,\,(\numM_i^c(M)-u_i^c)^+\big\}.
\]
A matching $M$ is called \emph{\diverse} if $\divv_1(M)=0$ (equivalently, $\divv_\infty(M)=0$). Among stable matchings, we say that $M$ is \emph{\utilopt} if it minimizes $\divv_1(M)$, and \emph{\egalopt} if it minimizes $\divv_\infty(M)$.
\end{definition}

The diversity objectives defined above arise as a special case of a more general framework based on \emph{institutional objective functions}.  
\begin{definition}[Institutional Objective Function]
For each institution $i \in I$, let $g_i : 2^{E(i)} \to \mathbb{R}$ be a set function that assigns a real-valued cost to any subset of students assigned to $i$. 
\end{definition}

Within this framework, we study the problem of finding a stable matching $M \in \mathcal{M}$ that minimizes either the utilitarian objective $\sum_{i \in I} g_i(M(i))$ or the egalitarian objective $\max_{i \in I} g_i(M(i))$. The diversity measures $\divv_1$ and $\divv_\infty$ correspond to specific choices of $g_i$. Similarly as for the above special case, among stable matchings, we say that $M$ is \emph{\utilopt} if it minimizes $\sum_{i\in I}g_i(M)$, and \emph{\egalopt} if it minimizes $\max_{i\in I}g_i(M)$.

% Here, $g_i$ assigns a score (or penalty) to a set of edges (or alternatively, the students that apply to $i$) incident to $i$; in other words, it scores the specific set of students assigned to institution $i$. We show that the structure of stable matchings surprisingly allows us to solve even this broad class of problems efficiently, whenever the functions $g_i$ are polynomial-time computable.
%
Each function $g_i$ assigns a numerical cost
to the subset of students matched to institution $i$. In this framework, $g_i$ serves as a quantitative measure of how well a local student composition aligns with the institution's specific distributional objectives.
To illustrate the versatility of this approach, we consider three distinct examples of distributional objective functions.

\paragraph{One-to-all counting diversity constraints.} In this setting, each student $s$ is associated with a subset of categories $c(s) \subseteq C$. Each institution $i$ imposes a lower bound $\ell_i^c$ and an upper bound $u_i^c$ for every category $c \in C$. When a student is assigned to institution $i$, they contribute to the count of \textbf{every} category they belong to. We define the distributional penalty function as:
\[
g_i(M) = \sum_{c \in C} \left( \max(0, \ell_i^c - \numM_i^c(M)) + \max(0, \numM_i^c(M) - u_i^c) \right)
\]

\paragraph{One-to-one counting diversity constraints.} In contrast, one-to-one counting requires that each student is assigned to \textbf{exactly one} of the categories for which they are eligible. Here, $g_i(M)$ represents the minimum total violation across all possible valid category assignments. Formally, let $S_i \subseteq S$ be the set of students matched to institution $i$ in $M$. We define $g_i(M)$ as the minimum over all type-assignment functions $\tau: S_i \to C$ (where $\tau(s) \in c(s)$) of:
\[
\sum_{c \in C} \left( \max(0, \ell_i^c - |\{ s \in S_i \mid \tau(s) = c \}|) + \max(0, |\{ s \in S_i \mid \tau(s) = c \}| - u_i^c) \right)
\]
In this case, the global objective $\sum_{i \in I} g_i(M)$ corresponds to the total diversity violation $\divv_1(M)$.

To compute $g_i(M)$ efficiently, we can reduce the problem to finding a maximum weight matching in a bipartite b-matching graph $H$.
For each category $c \in C$, we create two of vertices: $i_c$ (with capacity $\ell_i^c$) and $i_c'$ (with capacity $u_i^c - \ell_i^c$). We also include a ``slack'' vertex $i''$ with capacity $q_i$. On the other side of the graph, we create a vertex for each student $s \in S_i$. 

Edges are constructed as follows:
\begin{itemize}
    \item An edge between $s$ and $i_c$ exists if $c \in c(s)$, with weight \textbf{2} (prioritizing the satisfaction of lower bounds).
    \item An edge between $s$ and $i_c'$ exists if $c \in c(s)$, with weight \textbf{1} (accounting for students within the upper bound).
    \item An edge between $s$ and $i''$ exists for all $s$ with weight \textbf{0} (representing students that contribute to a violation of the upper bound).
\end{itemize}

\begin{claim}
    The maximum weight of an $S_i$-side complete matching in $H$ is $(|S_i| + \sum_{c \in C} \ell_i^c) - R$ if and only if $g_i(M)=R$.
\end{claim}
\begin{proof}
For the first direction, choose $R$ such that $(|S_i| + \sum_{c \in C} \ell_i^c) - R$ is the weight of a maximum weight $S_i$-side complete matching $M'$ in $H$. Create a category assignment $\tau$ such that $\tau (s)=c$ if $(s,i_c)\in M'$ or $(s,i_c')\in M'$.  For $s\in S_i$ such that $(s,i'')\in M$, let $c^*$ be a category such that $i_{c^*}$ and $i_{c^*}'$ are filled in $M'$ (if there is no such $c^*$, then $M'\setminus \{ (s,i'')\} \cup \{ (s,i_c'\}$ for some $c\in c(s)$ would have larger weight than $M'$, contradiction) and let $\tau (s) = c^*$. By the maximal weight property of $M'$, we can also assume that $i_c'$ has incident edges in $M'$ only if $i_c$ is fully filled and if $(s,i'')\in M'$, then $i_c$ and $i_c'$ are fully filled for all $c\in c(s)$.

We have that if $a_i=\sum_{c\in C}(\ell_i^c-|M'(i_c)|)^+$ and $b_i=|M'(i)|$, then the weight of $M'$ is $2\cdot (\sum_{c\in C}\ell_i^c - a_i)+ 0\cdot b_i + 1\cdot (|S_i|-\sum_{c\in C}\ell_i^c +a_i-b_i)=(|S_i|+\sum_{c\in C}\ell_i^c)-a_i-b_i$, so $a_i+b_i=R$. 
Furthermore, we also have that by the construction of $\tau$ and our observation that $i_c$ gets filled before $i_c'$, 
$\sum_{c \in C} ( \max(0, \ell_i^c - |\{ s \in S_i \mid \tau(s) = c \}|) + \max(0, |\{ s \in S_i \mid \tau(s) = c \}| - u_i^c)=a_i+b_i=R$. Thus, the value of $g_i(M)$ is at most $R$.

For the other direction, let $R$ be the value of $g_i(M)$, and let $\tau$ be the corresponding optimal category assignment. We construct a feasible matching $M'$ in $H$. For each institution $i$, let the students with $\tau (s) =c$ at $i$ in $M$ be $s_1^c\succ_i \cdots \succ_i s_{k_c}^c$. We assign the best $\ell_i^c$ of them to $i_c$, the next $u_i^c-\ell_i^c$ to $i_c'$ and the rest to $i''$ in $M'$. Clearly, no vertex capacity is violated.

For institution $i$, let $a_i$ denote $\sum_{c\in C}(\ell_i^c-|M'(i_c)|)^+$ and $b_i$ denote $|M'(i'')|$. Hence, we have $R=a_i+b_i$. The weight of the edges of $M'$ incident to copies of $i$ is then $2\cdot (\sum_{c\in C}\ell_i^c - a_i)+ 0\cdot b_i + 1\cdot (|S_i|-\sum_{c\in C}\ell_i^c +a_i-b_i)= (|S_i| + \sum_{c \in C} \ell_i^c) - R$, so the maximum weight of an $S_i$-side complete matching is at least this much, and by the first direction, it is exactly this much. 
\end{proof}

\paragraph{Matching siblings together.} In this scenario, the set of students is partitioned into a set of disjoint families $\Fam = \{f_1, \dots, f_k\}$, where each family $f_j$ consists of one or more siblings. The objective is to maximize the number of families whose members are all assigned to the same institution. 

To achieve this within our minimization framework, we define the distributional penalty for each institution $i$ as the negative of the number of families fully contained within that institution:
\[
g_i(M) = -|\{ f_j \in \Fam \mid f_j \subseteq M(i) \}|.
\]
Consequently, minimizing the global objective $\sum_{i \in I} g_i(M)$ is equivalent to maximizing the total number of families kept together across all institutions. In this sense, the function $g(\cdot)$ directly rewards the co-location of siblings.

\subsection{Structural Results for Stable Matchings}

In this section, we collect several structural results on stable matchings that will be used in the analysis of our algorithms. 
We begin with a classical property of stable matchings in many-to-many markets due to~\citet{baiou2000many} (see also~\cite{fleiner2003stable}).

% \zhnew{
% Roth and Sotomayor's Comparability Theorem, in the stable \(b\)-matching form given by Fleiner~\cite[Theorem~5]{fleiner2003stable}.
% \begin{theorem}[Comparability Theorem {\cite[Theorem~5]{fleiner2003stable}}]
% \label{fact:old}
% Let $M$ and $M'$ be two stable matchings and let $i$ be an institution.
% If $M(i)\neq M'(i)$, then
% $|M(i)|=|M'(i)|=q_i,$
% and the $q_i$ highest-priority students in
% $M(i)\cup M'(i)$ constitute either $M(i)$ or $M'(i)$.
% \end{theorem}
% }

% Or 
% \zhnew{
% \begin{theorem}[{\cite{baiou2000many}, Theorems 2--3}]
% \label{fact:old}
% Let $M$ and $M'$ be two stable matchings and let $i$ be an institution.
% If $M(i)\neq M'(i)$, then
% \[
%     |M(i)|=|M'(i)|=q_i.
% \]
% Moreover, the $q_i$ highest-priority students in
% $M(i)\cup M'(i)$ according to $\succ_i$ constitute either
% $M(i)$ or $M'(i)$.
% \end{theorem}
% }

\begin{theorem}[Comparability Theorem {\cite[Theorem~5]{fleiner2003stable}}]\label{thm:stable_diff}
For any two stable matchings $M,M'$ and any agent $i$, if $M(i)\ne M'(i)$, then 
\begin{itemize}
    \item $|M(i)|=|M'(i)|=q_i$,
    \item the set of $|M(i)\setminus M'(i)|$ worst agents of $M(i)\cup M'(i)$ according to $\succ_i$ are either $M(i)\setminus M'(i)$ or $M'(i)\setminus M(i)$,
    \item the $q_i$ best edges of $M(i)\cup M'(i)$ according to $\succ_i$ are either $M(i)$ or $M'(i)$.
\end{itemize}
\end{theorem}

The next result characterizes the student-optimal stable matching under additional constraints. In particular, it allows us to incorporate forbidden edges in the Deferred Acceptance algorithm (DA) via \texttt{BreakMatch} operations. 

\begin{theorem}[\cite{gusfield1987three}]\label{thm:forbiddenopt}
Let $F$ be a set of forbidden edges such that there exists a stable matching that avoids $F$. Then the student-proposing deferred acceptance algorithm, combined with \emph{\texttt{BreakMatch}} operations for edges in $F$, produces a stable matching $M_F$ that avoids $F$ and is weakly preferred by all students to any other stable matching that avoids $F$.
\end{theorem}

Finally, we recall the Rural Hospitals Theorem, which characterizes invariance properties across stable matchings.

\begin{theorem}[Rural Hospitals Theorem \citep{RoSo90a}]\label{thm:rural}
Let $M$ and $M'$ be any two stable matchings. Then, we have that:
\begin{enumerate}
\item  the set of matched residents is identical under $M$ and $M'$;
\item each hospital is assigned the same number of residents under both matchings; and
\item  any hospital that is undersubscribed in one stable matching is assigned the same set of residents in all stable matchings.
\end{enumerate}
\end{theorem}

\section{An algorithm for \egalopt\ stable matching with diversity constraints
% with overlapping categories
}\label{sec:DA-maxviol}

% In this setting, deciding if a maximally diverse matching exists is NP-hard (see e.g.. However, most interestingly, among stable matchings, we will show that even an $\ell_{\infty}$-most diverse stable matching can be found efficiently. Such an algorithm may be crucially useful for applications, where diversity constraints are important, but the structure of categories is arbitrary.

 % \gnote{Gergely: What is nice about this one is that it does not need a computation of a maximum weight stable matching, nor preprocessing, just a DA $\to$ keep as a warm-up and most practical algorithm.}
 
In this section, we study a model in which each student may belong to multiple, possibly overlapping categories. Although deciding whether a \diverse\ matching exists is NP-hard \citep{arnosti2024explainable}, we show that determining the existence of a \utilopt\ or \egalopt\ matching  among all stable matchings can be accomplished in polynomial time. This tractability stems from the fact that the search can be confined to the lattice of stable matchings, thereby substantially reducing the space of candidate solutions.

% For each institution $i\in I$ and category $c\in C$, let $u_i^{\ne c}=|M(i)|-\ell_i^c$ for an arbitrary stable matching $M$ ($|M(i)|$ is always the same) \zhnew{denote the upper bound of students who can be assigned to types other than $c$ under a maximally-diverse matching}. We say that $M$ \emph{satisfies} $u_i^{\ne c}$, if  $|\{ s\in M(i)\mid c\notin c(s)\}|\le u_i^{\ne c}$.

% \begin{claim}\label{claim:bounds}
%     If a stable matching $M$ satisfies all bounds $u_i^{\ne c}$ and $u_i^c$, then it is maximally-diverse.
% \end{claim}
% \begin{proof}
%     Suppose for the contrary that a lower bound $\ell_i^c$ is violated by a stable matching $M$. Then, we get that $|\{ s\in M(i)\mid c\notin c(s)\}|=|M(s)|-\numM_i^c(M)> |M(i)|-\ell_i^c$. Hence, $u_i^{\ne c}$ is violated by $M$, contradiction.
% \end{proof}

\subsection{Diversity-Constrained Deferred Acceptance}\label{sec:upper-bound-alg}

We next introduce the Diversity-Constrained Deferred Acceptance (DC-DA) algorithm to determine the existence of a stable matching that satisfies all diversity constraints. At a high level, the algorithm starts from a stable matching that ignores diversity constraints and iteratively modifies its edges in the preference graph, while preserving stability, until it either finds a stable and \diverse\ matching or concludes that no such matching exists.

% \paragraph{Transformation:} 
We first demonstrate that the existence of a stable matching satisfying lower bounds can be reduced to an equivalent problem involving only upper bounds. This reduction relies on the \emph{Rural Hospitals Theorem}, which guarantees that the number of students assigned to an institution $i$, denoted by $|M(i)|$, is invariant across all stable matchings. 
Specifically, for each institution $i \in I$ and category $c \in C$, we define a transformed upper bound:
\[
u_i^{\neq c} = |M(i)| - \ell_i^c
\]
This value represents the maximum number of students not belonging to category $c$ that institution $i$ can admit while still satisfying the original lower bound $\ell_i^c$. Formally, we say that a stable matching $M$ \emph{satisfies} the diversity upper bounds if for every $i \in I$ and $c \in C$:
\[
\left| \{ s \in M(i) \mid c \notin c(s) \} \right| \le u_i^{\neq c} \space \text{ and } \space |\{ s \in M(i) \mid c \in c(s) \}| \le u_i^{c}.
\]
This transformation simplifies the search process, as we can now treat the problem as finding a stable matching that does not exceed a set of maximum capacities for category and non-category members.

\begin{lemma}\label{lemma:DAFE:bounds}
    A stable matching $M$ is \diverse\ if and only if it satisfies all diversity upper bounds $u_i^c$ and $u_i^{\ne c}$.
\end{lemma}

\begin{proof}
    Suppose toward a contradiction that a stable matching $M$ satisfies $u_i^{\ne c}$ but violates the lower bound $\ell_i^c$. By definition, a violation of the lower bound implies $|\{\, s \in M(i) \mid c \in c(s) \,\}| < \ell_i^c$. It follows that:
    \[
    |\{\, s \in M(i) \mid c \notin c(s) \,\}| = |M(i)| - |\{\, s \in M(i) \mid c \in c(s) \,\}| > |M(i)| - \ell_i^c.
    \]
    Substituting the definition of $u_i^{\ne c}$, we get $|\{\, s \in M(i) \mid c \notin c(s) \,\}| > u_i^{\ne c}$, which contradicts the assumption that $M$ satisfies $u_i^{\ne c}$.
\end{proof}

We next present a modified Deferred Acceptance algorithm designed to compute a stable matching that respects the upper bounds $u_i^c$ and $u_i^{\neq c}$ for all institutions $i \in I$ and categories $c \in C$. The procedure maintains a dynamically updated set of \emph{forbidden edges} $F$, which acts as a structural constraint on the matching process. 

Whenever a diversity bound is violated at an institution, the algorithm identifies the offending edges, those associated with the least-preferred students contributing to the violation, and adds them to $F$. This effectively compels the institution to reject these applicants, triggering a chain of new proposals. By resuming the Deferred Acceptance process on this increasingly restricted preference graph, the algorithm monotonically traverses the stable matching lattice. The procedure continues until it either converges on a feasible stable matching that satisfies all bounds or exhausts the set of available edges, thereby proving that no such matching exists.

\paragraph{Algorithm DC-DA} 

\paragraph{Initialization:} We initialize the set of forbidden edges $F_0 = \emptyset$ and let $j$ index the iterations, starting with $j = 1$. First, we run the standard student-proposing Deferred Acceptance (DA) algorithm to obtain the student-optimal stable matching $M_0$. Let
\[
\rural = \{\, i \in I \mid |M_0(i)| < q_i \,\}
\]
denote the set of unfilled institutions, and let
\[
\rurals = \{\, s \in S \mid M_0(s) = \emptyset \,\}
\]
denote the set of unmatched students. By the Rural Hospitals Theorem, the sets of unmatched students and unfilled positions are invariant across all stable matchings. Therefore, if there exists some $i \in \rural$ that violates an upper bound $u_i^c$ or $u_i^{\ne c}$ in $M_0$, the algorithm terminates and outputs ``no solution.''
% By Rural theorem, the set of students matched to any undersubscribed institution $i$ remains the same in any stable matching.

\paragraph{Check Bounds:}
In each iteration $j$, we evaluate whether the current stable matching $M_{j-1}$ satisfies the upper bounds $u_i^c$ and $u_i^{\neq c}$ for all institutions $i \in I$ and categories $c \in C$. If all diversity and capacity constraints are met, the algorithm terminates and returns $M_{j-1}$ as a valid solution.

If a violation is detected, we update the forbidden edge set to force the institution to adjust its composition as follows.

If a bound $u_i^c$ is violated at institution $i$, let $s$ be the student of category $c$ who is least-preferred by $i$ in the current matching $M_{j-1}$. To reduce the count, we forbid $s$ and all students ranked lower than $s$ by $i$: 
% \zh{Since $s$ is the least preferred student who is matched to $i$ now, so none of the students in $s'$ needs to be removed.}\gnote{G: $s$ is only worst among category $c$. There can be students $s'$ from different categories who have to be rejected with $s$ (otherwise $s$ would block with $i$).}
    \[ F_j = F_{j-1} \cup \{ (s', i) \in S \times I \mid s \succeq_i s' \}. \]
Similarly, if a bound $u_i^{\neq c}$ is violated, let $s$ be the least-preferred student who does \emph{not} belong to category $c$. We update the forbidden set analogously:
    \[ F_j = F_{j-1} \cup \{ (s', i) \in S \times I \mid s \succeq_i s' \}. \]

After updating $F_j$, we increment $j$ and resume the student-proposing DA procedure on the reduced preference profile (any rejection that happened so far with forbidden edges $F_{j-1}\subset F_j$ must also happen with $F_j$ by Theorem~\ref{thm:forbiddenopt}).

\paragraph{Resuming DA:}
While there exist active matches that are now forbidden, i.e., $F_j \cap M_{j-1} \neq \emptyset$, we perform a \texttt{BreakMatch} operation for each $(s, i) \in F_j \cap M_{j-1}$ by removing $s$ from institution $i$ and leaving the student temporarily unmatched. The Deferred Acceptance algorithm then resumes from this partial matching. Displaced students continue their proposal sequence, moving down their preference lists while strictly bypassing any edges currently in $F_j$. 

The algorithm monitors for two specific failure conditions that signal the non-existence of a feasible stable matching:
\begin{itemize}
\item \textbf{Displacement of Invariant Students:} If a student $s \notin \rurals$ (one who was matched in the initial student-optimal matching) exhausts their preference list without finding a valid assignment.
\item \textbf{Violation of Unfilled Positions:} If any student proposes to an institution $i \in \rural$, effectively attempting to alter the invariant occupancy of an under-subscribed institution.
\end{itemize}

If either condition is met, the algorithm terminates and outputs ``no solution,'' as the Rural Hospitals Theorem implies no stable matching can satisfy the given bounds. Otherwise, the procedure converges to a new stable matching $M_j$. We then increment $j$ and return to the bound-checking step to ensure all diversity constraints are satisfied.

\begin{theorem}\label{thm:DC-DA}
The DC-DA algorithm terminates in $\mathcal{O}(|E|)$ time,  where $E \subseteq S \times I$ denotes the set of edges in the preference graph. It outputs a \diverse\ stable matching if one exists, and outputs ``no solution'' otherwise. Furthermore, the resulting matching is the student-optimal stable matching within the set of all \diverse\ stable matchings. 
\end{theorem}

We prove Theorem~\ref{thm:DC-DA} through the following lemmas. First, Lemma~\ref{lemma:DC-DA:stable} establishes that any matching returned by the DC-DA algorithm is a \diverse\ stable matching. Then, Lemma~\ref{lemma:DC-DA:termination} proves the algorithm's completeness: if DC-DA outputs ``no solution,'' then no \diverse\ stable matching exists and also that otherwise the output is student-optimal. Finally, Lemma~\ref{lemma:DC-DA:running_time} shows that the procedure executes in linear time relative to the number of edges in the preference graph. The details of the proof are presented in Section~\ref{sec:proof:DC-DA}.

% \begin{remark}
%     To find an \egalopt\ stable matching, we just use the above algorithm, and if there is no \diverse\ stable matching, then we relax each bound $u_i^c$ and $u_i^{\ne c}$ by 1, until a \diverse\ stable matching is found. We are guaranteed to find one in at most $\max_{i\in I}q_i$ such iterations. Furthermore, by Corollary~\ref{cor:stuopt}, the output will be a student optimal among the $\ell_\infty$-most diverse stable matchings. 
% \end{remark}

Next, we show how the DC-DA can be used to find an \egalopt\ stable matching.

\begin{theorem}\label{thm:DC-DA-most-diverse}
A student-optimal stable matching among the \egalopt\ ones can be identified in $\mathcal{O}(|E||S|)$ time.
\end{theorem}

% To identify an egalitarian or nearly-diverse stable matching, we apply the DC-DA algorithm within an iterative relaxation framework. If no stable matching satisfies the initial diversity bounds, we incrementally relax each constraint $u_i^c$ and $u_i^{\neq c}$ by one unit per iteration. This process continues until the DC-DA algorithm identifies a feasible stable matching. We are guaranteed to find such a matching in at most $\max_{i \in I} q_i$ iterations, as the constraints eventually become non-binding. Furthermore, this procedure ensures that the resulting output is the student-optimal stable matching among all those that are $\ell_\infty$-most diverse---effectively minimizing the maximum deviation from the original diversity targets.
\begin{proof}
To identify an \egalopt\ stable matching, we apply the DC-DA algorithm within an iterative relaxation framework. If no stable matching satisfies the initial diversity bounds, we incrementally relax each constraint $u_i^c$ and $u_i^{\neq c}$ at the same time by one unit per iteration. 
Formally, the algorithm seeks the smallest $\delta^*$ such that there exists a stable matching $M$ satisfying:
\[
| \{s \in M(i) \mid c(s) = c\} | \le u_i^c + \delta^* \quad \text{and} \quad | \{s \in M(i) \mid c \notin c(s)\} | \le u_i^{\neq c} + \delta^*
\]
for all $i \in I$ and $c \in C$. This objective is equivalent to minimizing the maximum violation across all diversity constraints, which precisely defines the \egalopt\ stable matching. 
Since DC-DA (by Theorem~\ref{thm:DC-DA}) finds the student-optimal matching for any fixed set of bounds, and the set of stable matchings satisfies the lattice property, the first $\delta$ for which DC-DA returns a solution yields the student-optimal matching among all matchings with the same minimum worst diversity violation.

The procedure iterates on $\delta$. In each iteration, we run the DC-DA algorithm which takes $\mathcal{O}(|E|)$ time. The maximum number of iterations is bounded by the maximum capacity of any institution, $\max_{i \in I} q_i$, as the constraints become non-binding when $\delta \ge q_i$. Therefore, the total running time is $\mathcal{O}(|E| \cdot \max_{i \in I} q_i)\le \mathcal{O}(|E||S|)$, which is polynomial in the size of the input.
\end{proof}

%\begin{lemma}
%    The maximally-diverse stable matchings form a lattice.
%\end{lemma}
%\begin{proof}
%    By Theorem~\ref{fact:stable_diff}, for any two stable matchings $M,M'$, we have that if an institution chooses the best $|M(i)|$ or the worst $|M(i)|$ students from the set $M(i)\cup M'(i)$, then it is either $M(i)$ or $M'(i)$. Hence, if $M(i)$ and $M'(i)$ are both \diverse\, then so must be their meet $M\vee M'$ (where all institutions $i$ choose the best $|M(i)|$ students simultaneously) and join $M\wedge M'$ (where all institutions $i$ choose the worst $|M(i)|$ students simultaneously). Hence, they form a lattice. 
%\end{proof}
%\begin{corollary}
%    The $\ell_\infty$-most diverse stable matchings form a lattice.
%\end{corollary}

\subsection{Proof of Theorem~\ref{thm:DC-DA}}
\label{sec:proof:DC-DA}

We first show that whenever the DC-DA algorithm terminates and returns a matching, that matching is a \diverse\ stable one.

\begin{lemma}\label{lemma:DC-DA:stable}
Suppose DC-DA terminates at iteration $T$ with a matching $M_T$. For every $k \in \{0, \dots, T\}$, the provisional matching $M_k$ is stable with respect to the original preference profiles $\succ_S$ and $\succ_I$ and $M_T$ is a \diverse\ stable matching.
\end{lemma}

 \begin{proof}
Let $M_k$ denote the provisional matching obtained after the $k$-th iteration of the algorithm, with $M_0$ being the initial student-optimal stable matching. We first claim that institutions weakly improve between iterations in the DC-DA algorithm.

\begin{claim}\label{claim:inst-impr}
Suppose DC-DA reaches a matching $M_k$ for $k\ge 1$. 
    For any institution $i \in I$, we have $|M_k(i)| = |M_{k-1}(i)|$. Furthermore, for any institution $i$, it prefers its least-preferred student in $M_k$ to its least-preferred student in $M_{k-1}$ according to $\succ_i$.
\end{claim}

\begin{proof}
Since the algorithm does not halt at iteration $k$, by construction, no student who was matched in $M_0$ (i.e., $s \notin \rurals$) exhausts their preference list, and no student proposes to an institution in $\rural$. This implies that all students who were matched in $M_{k-1}$ remain matched in $M_k$ and others have already exhausted their lists by $M_0$. 

Consequently, any institution $i$ that was at full capacity in $M_{k-1}$ (the ones not in $\rural$) must remain at full capacity in $M_k$. As long as $i$ remains at full capacity, by the fundamental rejection rule of the DA, its least preferred student weakly improves. When a \texttt{BreakMatch} operation is performed on a forbidden edge $(s,i)$ of $i$, then no student weakly worse than $s$ ever proposes to $i$ (the edge gets forbidden), so $i$ must reach $M_k$ filled with students strictly better than $s$, which is weakly preferred to its least-preferred student in $M_{k-1}$. Thus, the rank of its worst-matched student must weakly improve.

Furthermore, for any institution $i \in \rural$, the assumption that the algorithm does not halt implies $i$ receives no new proposals during this iteration. Therefore, its set of assigned students remains unchanged, $M_k(i) = M_{k-1}(i)$. For institutions not in $\rural$, they remain at capacity by the logic above. Thus, $|M_k(i)| = |M_{k-1}(i)|$ for all $i \in I$, and the preference bound holds.
\end{proof}

 We claim that $M_k$ is stable.
Suppose, for the sake of contradiction, that there exists a blocking pair $(s, i)$ for the matching $M_k$ with respect to the original preferences.

First, consider the case where $(s, i) \notin F_k$. If $(s, i)$ is not a forbidden edge, then $s$ must have proposed to $i$ and been rejected at some point during the student-proposing DA process (so $i\notin \rural$ is filled as the algorithm reached $M_k$). Under the standard DA mechanism, an institution only rejects a student if it is already matched with students it prefers more. By Claim~\ref{claim:inst-impr}, the rank of the least-preferred student assigned to $i$ weakly improves across iterations. Consequently, in the final matching $M_k$, institution $i$ still prefers every student in $M_k(i)$ to $s$. Thus, $(s, i)$ cannot form a blocking pair, contradiction.

Next, consider the case where $(s, i) \in F_k$. By the algorithm's construction, an edge is only added to the forbidden set if the institution $i$ is not in $\rural$. Since $i \notin \rural$, it was at full capacity when $(s, i)$ was forbidden and $i$ kept only students it strictly preferred to $s$ after the \texttt{BreakMatch} operations, and weakly worse ones became forbidden. Since the rank of the least-preferred student at $i$ weakly improves in each subsequent iteration (Claim~\ref{claim:inst-impr}), $i$ remains matched with a full set of students in $M_k$ that it prefers to $s$. Therefore, $i$ would not deviate to match with $s$, and $(s, i)$ cannot be a blocking pair, contradiction again.

Since the algorithm terminated with $M_T$, it must satisfy all diversity upper bounds $u_i^c$ and $u_i^{\ne c}$, hence by Lemma~\ref{lemma:DAFE:bounds} it is \diverse.
\end{proof}

%\begin{lemma}\label{lemma:DC-DA:maximally_diverse}
%If the DC-DA algorithm terminates and outputs a matching $M^*$, then $M^*$ is a \diverse\ stable matching.
%\end{lemma}

We next establish the completeness of the algorithm by demonstrating that an unsuccessful termination serves as a formal certificate of non-existence. The proof relies on the fact that any edge $(s, i)$ that is either added to the forbidden set $F = \bigcup_j F_j$ or rejected during the Deferred Acceptance procedure cannot belong to any \diverse\ stable matching. 

\begin{lemma}\label{lemma:DC-DA:termination}
If the DC-DA algorithm terminates and outputs ``no solution,'' then there exists no stable and \diverse\ matching for the given instance. Otherwise, the output $M_T$ is student-optimal among \diverse\ stable matchings.
\end{lemma}

% \vspace{0.5em}
% \noindent\textbf{3. Termination and "No Solution" Conditions.} \\
% By the Rural Hospitals Theorem (Theorem \ref{thm:rural}), the set of matched students, the set of matched institutions, and the number of positions filled at each institution are identical across \emph{all} stable matchings. 

% If the algorithm halts because an institution $i \in \rural$ violates a diversity bound in $M_0$, then by the Rural Hospitals Theorem, $i$ receives the exact same subset of students in every stable matching. Thus, the bound is violated in all stable matchings, and no MDSM exists.

% If the algorithm halts because a student $s \notin \rurals$ exhausts their preference list, or applies to an institution in $\rural$, this implies $s$ was forced to propose to an institution worse than their assignment in the student-pessimal stable matching, or a rural hospital received a new proposal. Suppose the output is incorrect and there exists an MDSM. Then, the first option contradicts Theorem~\ref{thm:rural}. The second option would mean that there is a stable matching, where a student is worse off than in the student-pessimal stable matching, also a contradiction.

\begin{proof}

Let $\mathcal{M}^{MD}$ denote the set of \diverse\ stable matchings.
To prove this lemma, we show by induction that for any $k\le T$ (where $T$ is the number of iterations), no \diverse\ stable matching  $M\in \mathcal{M}^{MD}$ intersects $F_k$. For $k=0$, $F_0=\emptyset$, so this is trivial.

Let $\mathcal{M}(F_k)$ denote the set of stable matchings in the original instance that do not contain any edge in $F_k$. %By Theorem~\ref{thm:forbiddenopt}, if $\mathcal{M}(F_k)\ne \emptyset$, then a matching $M_k$ obtained after iteration $k$ and it is the student-optimal stable matching within $\mathcal{M}(F_k)$.

\textbf{Inductive step:} Assume that no matching $M \in \mathcal{M}^{MD}$ contains any edge in $F_{k-1}$. Let $M^*$ be an arbitrary matching in $\mathcal{M}^{MD}$. By the inductive hypothesis, $M^* \in \mathcal{M}(F_{k-1})$. Since $M_{k-1}$ is student-optimal in $\mathcal{M}(F_{k-1})$, it follows that $M_{k-1}(s) \succeq_s M^*(s)$ for all students $s$.

Suppose that at iteration $k$, the algorithm adds the edges $\{(s', i) \mid s \succeq_i s'\}$ to $F_k$ because a bound $u_i^c$ is violated at institution $i$ in $M_{k-1}$. Here, $s$ is the least-preferred category $c$ student in $M_{k-1}(i)$. Assume, for contradiction, that $M^*$ contains some edge $(s', i)$ where $s \succeq_i s'$. Because $M^*$ satisfies $u_i^c$ while $M_{k-1}$ violates it, and because both matchings must assign the same total number of students to $i$ (by the Rural Hospitals Theorem), $M_{k-1}(i)$ must contain at least one category $c$ student $s''$ that is not present in $M^*(i)$.

Since $s'' \in M_{k-1}(i)$ and $s'' \notin M^*(i)$, the student-optimality of $M_{k-1}$ in $\mathcal{M}(F_{k-1})$ implies $i = M_{k-1}(s'') \succ_{s''} M^*(s'')$. Furthermore, since $s$ is the least-preferred category $c$ student in $M_{k-1}(i)$, we have $s'' \succeq_i s$. Combined with our assumption that $s \succeq_i s'$, we get $s'' \succeq_i s'$. As $s'' \neq s'$ ($s'\in M^*(i)$, but $s''\notin M^*(i)$), we get that $i$ strictly prefers $s''$ to $s'$. Thus, $(s'', i)$ forms a blocking pair for $M^*$, contradicting its stability. A symmetric argument holds for violations of $u_i^{\ne c}$. 

Therefore, no \diverse\ stable matching contains any edge in $F_k$. %Since the DC-DA algorithm---via the Deferred Acceptance procedure---systematically finds the student-optimal matching in $\mathcal{M}(F_k)$ (if one exists), it follows that if the algorithm reaches a matching $M_k$, then $M_k$ is the student-optimal matching in $\mathcal{M}(F_k)$ .

By the Rural Hospitals Theorem (Theorem~\ref{thm:rural}), the set of matched students, the set of unfilled institutions, and the number of students assigned to each institution are invariant across all stable matchings.

First, suppose the algorithm terminates because an institution $i \in \rural$ violates an upper bound in $M_0$. By Theorem~\ref{thm:rural}, institution $i$ must be matched with the same set of students in every stable matching. Since the bound is violated in $M_0$, it is necessarily violated in all stable matchings, and thus no \diverse\ stable matching exists. 

Next, suppose the algorithm terminates because a student $s \notin \rurals$ exhausts their preference list or proposes to an institution $i \in \rural$ (which is thus not $M_0(s)$). Assume, for contradiction, that a \diverse\ stable matching $M^*$ exists. We know that $M^*$ cannot contain any edges from the forbidden set $F_T$. Consequently, $M^*\in \mathcal{M}(F_T)$, so $\mathcal{M}(F_T)\ne \emptyset$ and thus by Theorem~\ref{thm:forbiddenopt} the DA  with \texttt{BreakMatch} for the forbidden edge set $F_T$  finds a stable matching $M_T$. Thus, this DA cannot have $s\notin S$ exhausting their preference list or proposing to some $i\in \rural$ if $i\ne M_0(s)$ by the Rural Hospital Theorem. Since any rejection in the DC-DA up to iteration $T$ must happen in this DA with forbidden edge set $F_T$ too, we get a contradiction.

Finally, if the algorithm terminates with a matching $M_T$, then $M_T\in \stabset (F_T)$, and we have showen that any \diverse\ stable matching is also in $\stabset (F_T)$ and by Theorem~\ref{thm:forbiddenopt}, $M_T$ is student-optimal in $\stabset (F_T)$ as desired (using again that any rejection in the DC-DA up to iteration $T$ must happen in the DA with \texttt{BreakMatch} for the forbidden edge set $F_T$ too).
\end{proof}

% \vspace{0.5em}
% \noindent\textbf{4. Running Time.} \\
% The algorithm maintains the state of the Deferred Acceptance procedure. When rejections are forced, the algorithm resumes precisely from the current state. Consequently, each student proposes to each institution on their preference list at most once. Evaluating bounds and forcing rejections takes time proportional to the degree of the institution. Thus, the total running time is bounded by the number of preferences, which is $\mathcal{O}(|E|)$.
We next analyze the computational complexity of the DC-DA algorithm, demonstrating that it achieves optimal linear time complexity relative to the size of the preference graph.

\begin{lemma}\label{lemma:DC-DA:running_time}
The DC-DA algorithm terminates in $\mathcal{O}(|E|)$ time.
\end{lemma}

\begin{proof}
The algorithm's efficiency stems from the fact that it maintains the state of the Deferred Acceptance (DA) process across iterations. By resuming proposals from the last rejected student rather than restarting from scratch, the algorithm ensures that each student traverses their preference list at most once. 

Each proposal and subsequent acceptance or rejection by an institution takes constant time, assuming standard data structures. When a diversity bound is violated, identifying the least-preferred student $s$ (within or without category $c$) and updating the forbidden set $F$ can be done efficiently by maintaining two additional sorted lists of current matches for each institution $i$ and each category $c$. Since each edge $(s, i)$ is added to $F$ at most once and triggers at most one \texttt{BreakMatch} operation, the total work for updating $F$ is $\mathcal{O}(|E|)$. Finally, by maintaining counters for each category $c$ at each institution, bound checking also takes constant time per proposal. Summing these components, the total running time is linear in the number of edges.
\end{proof}

%\begin{lemma}\label{lemma:DC-DA:student-optimal}
%The matching $M^*$ produced by the DC-DA algorithm is the student-optimal stable matching within the set of all stable and \diverse\ matchings.
%\end{lemma}

%\begin{proof}
  %  In Lemma~\ref{lemma:DC-DA:termination} we have already shown that if $F$ is the set of forbidden edges at the end of the algorithm, then no \diverse\ stable matching contains edges from $F$. Hence, by Theorem~\ref{thm:forbiddenopt}, the output is student-optimal among \diverse\ stable matchings. 
%\end{proof}

% \begin{corollary}\label{cor:stuopt}
%     If the algorithm outputs a matching $M$ then it is a student-optimal stable matching among the \diverse\ stable matchings. 
% \end{corollary}
% \begin{proof}
%     In Theorem~\ref{thm:DC-DA} we have already shown that if $F$ is the set of forbidden edges at the end of the algorithm, then no \diverse\ stable matching contains edges from $F$. Hence, by Claim~\ref{lemma:stuopt}, the output is student-optimal among \diverse\ stable matchings. 
% \end{proof}

\section{The General Algorithm}\label{sec:main}

In this section, we address the general matching problem by incorporating arbitrary institutional objectives where an institution's goals extend beyond simple diversity requirements; instead, they are characterized by a distributional objective function $g_i$ defined over its set of matched students. Formally, given a set of distributional objective functions $g_i: 2^{E(i)} \to \mathbb{R}$ for each institution $i \in I$, our goal is to identify a stable matching $M$ that minimizes the aggregate institutional cost:
\[
g(M) := \sum_{i \in I} g_i(M(i)).
\]

% Before diving into the mechanics, we briefly outline the main intuition behind the algorithm. While the total number of stable matchings in an instance can be exponential, the number of distinct student sets assigned to any single institution $i$ across \emph{all} stable matchings is surprisingly small (bounded linearly by the number of acceptable students). Because of this property, we can perfectly map the complex, non-linear set function $g_i$ to a simple linear edge-weight function $\w$ that behaves identically on the restricted domain of stable matchings. The problem then elegantly reduces to finding a minimum-weight stable matching.

Although the total number of stable matchings in a given instance can be exponential, the number of distinct student sets that any single institution $i$ can be assigned across all stable matchings is remarkably small. Specifically, it is bounded linearly by the number of students acceptable to $i$.

Exploiting this structural sparsity, we can linearize the complex, non-linear set function $g_i$. We construct a simple edge-weight function $\w$ that, while simple in form, behaves identically to $g_i$ over the restricted domain of stable matchings, thus rendering the problem efficiently solvable. This result is particularly significant, as it reveals that global institutional welfare can be maximized---or aggregate penalties minimized---without compromising the fundamental requirement of stability.

\subsection{Structural Preliminaries}

% Given a set of edges $H$, let $\top_i(H)$ denote the most preferred student of institution $i$ according to $\succ_i$ in the assignment $H(i)=H\cap E(i)$, and $\bottom_i(H)$ the least preferred student. Because an institution might see the exact same subset of students across several different stable matchings, we group these into distinct sets.

Given a set of edges $F$ of $E(i)$, let $\bottom_i(F)$ denote the least-preferred (cutoff) student of $i$. While an institution may be assigned the same set of students across various stable matchings, we are interested in the distinct collections of students that can be stably matched to $i$.

 Stable sets do not form a chain by inclusion. Instead, they form a chain ordered by preference: their worst students are totally ordered. This implies that as we transition between stable matchings, an institution’s admission threshold moves monotonically.

By Theorem~\ref{thm:stable_diff}, distinct, nonempty stable sets must have distinct cutoff students (i.e., if $\emptyset\ne E_i^x \neq E_i^y\ne \emptyset$, then $\bottom(E_i^x)\ne \bottom(E_i^y)$). Importantly, this also implies $k_i\le |E_i|$, whenever stable sets are nonempty.
Here $k_i$ represents the number of distinct stable sets for institution $i$.

\begin{corollary}
\label{cor:stab-sets}
If institution $i$ is assigned at least one student in a
stable matching, then its distinct stable sets can be
indexed as $E_i^1,\ldots,E_i^{k_i}$ such that
\[
    \operatorname{bottom}_i(E_i^1)
    \succ_i \cdots \succ_i
    \operatorname{bottom}_i(E_i^{k_i}).
\]
Consequently, $k_i\leq |E(i)|$.
If institution $i$ is unmatched in every stable matching,
then its only stable set is $\varnothing$.
\end{corollary}

%\begin{corollary}\label{cor:stabsets}
%We can index the stable sets $E_i^1,\dots, E_i^{k_1}$ of $i$ such that their cutoff students are strictly decreasing in preference: $\bottom_i(E_i^1) \succ_i \dots \succ_i \bottom_i(E_i^{k_i})$. Most importantly, when stable sets are nonempty, then $k_i\le |E(i)|$. 
%\end{corollary}

% The following lemma formalizes the bounds on these sets, showing that the number of stable sets per institution is strictly limited, and that each set can be uniquely identified by its best or worst student.

% \begin{lemma}\label{lem:stabsetnum}For any institution $i$, the number of possible stable sets is bounded by the number of acceptable students, i.e., $k_i\le |E(i)|$.
% \end{lemma}

\subsection{The Algorithm}

With the structure of stable sets established, we now present our algorithm, which operates in three distinct phases.

\paragraph{Step 1: Compute Stable Sets.} For each institution $i \in I$, we compute its complete collection of stable sets $E_i^1, \dots, E_i^{k_i}$. Given an institution $i$ with preferences $\succ_i = s_1 \succ_i \dots \succ_i s_{|E(i)|}$, we identify these sets using the following iterative procedure:

\begin{enumerate}
    \item \textbf{Initialization:} Execute the student-proposing Deferred Acceptance (DA) algorithm to obtain the student-optimal (and thus institution-pessimal) stable matching $M_0$. Set $j=0$ and initialize the forbidden edge set $F_0 = \emptyset$.
    
    \item \textbf{Empty Case:} If $M_0(i) = \emptyset$, the Rural Hospitals Theorem implies this is the only stable set for $i$; we terminate the process for this institution.
    
    \item \textbf{Iterative Discovery:} If $M_0(i) \neq \emptyset$, we find subsequent stable sets by forbidding the current worst student and resuming DA:
    \begin{itemize}
        \item Identify the least-preferred student at $i$ in the current matching: $s_{j+1}' = \bottom_i(M_j)$.
        \item Update the forbidden edge set: $F_{j+1} = F_j \cup \{ (s_{j+1}', i) \}$.
        \item Continue the DA process from the current state with the updated set $F_{j+1}$ to find the next stable matching $M_{j+1}$. (Note: Since $F_j \subseteq F_{j+1}$, we resume the algorithm rather than restarting.)
        \item Increment $j$.
    \end{itemize}
    
    \item \textbf{Termination:} The process stops when no stable matching $M_{j+1}$ exists under $F_{j+1}$. The distinct sets $M_0(i), \dots, M_j(i)$ constitute the stable sets for $i$, ordered such that their most-preferred students get more and more preferred by $i$.
\end{enumerate}

\paragraph{Step 2: Linearize Objective Functions into Edge Weights.} Let the stable sets of $i$ be $E_i^1,\dots, E_i^{k_i}$.
We transform the set evaluations $g_i(E_i^j)$ into linear edge weights $w(e)$.

Since every stable set is uniquely identified by its cutoff student, we use the cutoff edge as the representative edge for that stable set. The weight assigned to this edge serves only to encode the value of the corresponding stable set and has no relation to the institution’s priority ordering.

If $E_i^1 = \emptyset$, every stable matching avoids $i$, and $g_i(M)$ is simply $g_i(\emptyset)$ for all stable matchings and it is equivalent to minimize $g(M)-g_i(M)$. Otherwise, we exploit the fact that each stable set $E_i^j$ is uniquely identified by its cutoff student $s_j = \bottom_i(E_i^j)$. Let $\bottom_i(E_i^1) \succ_i \dots \succ_i \bottom_i(E_i^{k_i})$.

We assign $w(e) = 0$ for all edges that do not contain a cutoff student. For the cutoff students, we define weights recursively for $j = 1, \dots, k_i$:
\[
\w(s_j, i) = g_i(E_i^j) - \sum_{e \in E_i^j \setminus \{(s_j, i)\}} w(e)
\]
This recursive definition ensures that for any stable set $E_i^j$, the sum of the weights of its constituent edges exactly matches the function evaluation: $\sum_{e \in E_i^j} \w(e) = g_i(E_i^j)$.

\paragraph{Step 3: Minimum Weight Stable Matching.}
% Finally, we utilize the edge weights from Step 2 to solve the global problem. Since any stable matching $M$ must restrict to exactly one stable set $E_i^j$ for each institution $i$, the total weight of the matching corresponds perfectly to our original objective: $w(M) = \sum_{i \in I} g_i(M(i))$. 

Finally, we utilize the edge weights constructed in Step~2 to solve the global optimization problem. Step~2 establishes an exact correspondence between the original institutional objective and the edge-weight objective over the domain of stable matchings. Specifically, in any stable matching $M$, each institution $i$ is assigned exactly one of its stable sets $E_i^j$. By construction of the edge weights, the total weight of the edges incident to $i$ is exactly $g_i(E_i^j)=g_i(M(i))$. Therefore,
\[
w(M)=\sum_{e\in M}w(e)=\sum_{i\in I}g_i(M(i)),
\]
showing that minimizing the total edge weight is equivalent to minimizing the original aggregate institutional objective.

We find the minimum $w$-weight stable matching by employing efficient algorithms for the many-to-one stable matching polytope or the rotation digraph method~\citep{irving1987efficient, eirinakis2012finding}. The overall complexity is $\mathcal{O}(|S \cup I|^2 (|S \cup I| \log^2 |S \cup I| + T_g))$, where $T_g$ is the time required to evaluate the functions $g_i$ on the stable sets.

\subsection{Correctness of the Algorithm}

% We next show that the algorithm correctly computes a stable matching M that minimizes the aggregate
% institutional cost.

We now demonstrate that the algorithm correctly identifies a stable matching $M$ that minimizes the aggregate institutional cost.

% \begin{claim}\label{claim:step1} Step 1 successfully computes all possible stable sets of each $i\in I$.
% \end{claim}
% \begin{proof}
% By Theorem~\ref{fact:stable_diff}, the initial DA algorithm finds the stable set of $i$ containing the absolute worst possible cut-off student across all stable matchings. By Lemma~\ref{lem:stabsetnum}, this specific student is not included in any other stable set, since she would need to be the cut-off student in that set too. Therefore, by forbidding this student's edge and running DA (with BreakMatch for forbidden edges) again, we successfully find the stable set of $i$ with the second-worst cut-off student by Theorem~\ref{thm:forbiddenopt}. This second student similarly cannot be included in any remaining stable sets, allowing us to iteratively uncover every stable set without skipping any.
% \end{proof}

\begin{claim}\label{claim:step1} 
Step 1 successfully identifies the complete collection of stable sets for each institution $i \in I$.
\end{claim}
\begin{proof}
By the properties of the stable matching lattice, the initial student-proposing DA algorithm identifies the stable matching $M_0$ that is institution-pessimal. This matching $M_0(i)$ corresponds to the stable set of $i$ with the least-preferred possible cutoff student across all stable matchings. By Corollary~\ref{cor:stab-sets}, this cutoff student $s_0$ is unique to $M_0(i)$ and cannot serve as the cutoff for any other stable set. 

By forbidding the edge associated with this cutoff student and continuing the DA procedure, we effectively remove $M_0$ from the feasible set, but no stable matching $M$ with $M(i)\ne M_0(i)$. According to Theorem~\ref{thm:forbiddenopt}, each subsequent iteration $k$ identifies the institution-pessimal stable matching among the remaining set of stable matchings with a cutoff student $s_k$. By Corollary~\ref{cor:stab-sets}, this cutoff student cannot serve as the cutoff for any other stable set not yet identified, as it is the worst possible student for $i$ apart from $s_0,\dots, s_{k-1}$, for whom we have identified all possible stable sets containing them by induction. 

By repeating this process, we iteratively ``sweep'' upward through the lattice, uncovering each distinct stable set in increasing order of institutional preference. 
Since the number of such sets is finite and strictly ordered by their cutoff students, this procedure uncovers every stable set without omission.
\end{proof}

\begin{claim}\label{claim:linearize}
For any institution $i \in I$ and any stable set $E_i^j$, the constructed weight function $w$ satisfies:
\[ \sum_{e \in E_i^j} w(e) = g_i(E_i^j) \]
\end{claim}

\begin{proof}
We proceed by induction on $j$, the index of the stable sets as ordered from institution-optimal to institution-pessimal ($j = 1, \dots, k_i$).

\textbf{Base Case ($j=1$):} 
Consider the institution-optimal stable set $E_i^1$ and its cutoff student $s_1 = \bottom_i(E_i^1)$. By the recursive definition in Step 2:
\[ w(s_1, i) = g_i(E_i^1) - \sum_{e \in E_i^1 \setminus \{(s_1, i)\}} w(e) \]
Rearranging terms, we obtain $\sum_{e \in E_i^1} w(e) = g_i(E_i^1)$. For any edge $e \in E_i^1$ where $e$ is not the cutoff $(s_1, i)$, the weight $w(e)$ is $0$ because such an edge cannot be a cutoff student for any set more preferred than the institution-optimal set.

\textbf{Inductive Step:} 
Assume the claim holds for all stable sets $E_i^1, \dots, E_i^{j-1}$. Consider the $j$-th stable set $E_i^j$ and its unique cutoff student $s_j = \bottom_i(E_i^j)$. Our construction specifies:
\[ w(s_j, i) = g_i(E_i^j) - \sum_{e \in E_i^j \setminus \{(s_j, i)\}} w(e) \implies \sum_{e \in E_i^j} w(e) = g_i(E_i^j) \]
The validity of this construction relies on the property that any edge $e \in E_i^j$ other than the current cutoff $(s_j, i)$ must either have a weight of $0$ or be a cutoff student $s_m$ for some stable set $E_i^m$. By Corollary~\ref{cor:stab-sets}, if $s_m \in E_i^j$ for $m \neq j$, then $s_m$ must be strictly more preferred by institution $i$ than $s_j$. Because we process stable sets in order of decreasing institutional preference, the weight $w(s_m, i)$ was uniquely determined in a previous step $m < j$. Since the definition of $w(s_j, i)$ accounts for the cumulative weights of all such preceding cutoffs present in $E_i^j$, the sum is guaranteed to equal $g_i(E_i^j)$.
\end{proof}

\begin{theorem}\label{thm:main}A stable matching $M$ that minimizes $\sum_{i\in I}g_i(M)$ can be found in $O(|S\cup I|^2(|S\cup I|\log^2|S\cup I|+ T_g))$, where $T_g$ is the time to evaluate the functions $g_i$ on the stable sets ($T_g$ is polynomial, whenever the $g_i$ objectives are efficiently computable, as in our examples). Furthermore, we can define edge weights $\w(e)$ such that the optimal stable matchings correspond exactly to minimum $\w$-weight stable matchings. 
\end{theorem}
\begin{proof} Step 1 runs a DA with increasing the sets of forbidden edges in a way such that proposals happen at most once, so runs in $O(|E|)$ time per institution, so finding all stable sets in $O(|E|\cdot |I|)\le O(|S\cup I|^3)$. 

Step 2 solves a recursive formula that assigns valid edge weights in at most $|E(i)|$ steps, each requiring a single evaluation of $g_i$, in total requiring $O(|E|\cdot T_g)=O(|S\cup I|^2\cdot T_g)$ time. 

For Step 3, we can use the algorithm of \cite{eirinakis2012finding} that runs in $O(|S\cup I|^3\log^2|S\cup I|)$ time to find a minimum weight stable matching. %Finally, by Theorem~\ref{thm:rural}, we can assume that for any institution $i$, $q_i\le |M_0(i)|+1$, where $M_0$ is the student optimal stable matching (otherwise we can decrease to $|M(i)|+1$ without affecting the set of stable matchings). Hence, we also get that $|I'|\le |I|+|M|\le |I\cup S|$ and therefore $O(|S\cup I'|^3\log^2|S\cup I'|)=O(|S\cup I|^3\log^2|S\cup I|)$.

Summarizing the running times of all steps we get $O(|S\cup I|^2(|S\cup I|\log^2|S\cup I|+ T_g))$ as claimed.
\end{proof}

\begin{remark}\label{rem:maxcase}To find a stable matching that minimizes the egalitarian objective $\max_{i\in I}g_i(M)$, we can use a thresholding approach. Let $R_1\le R_2\le \dots \le R_k$ be the distinct values of $g_i(E_i^j)$ across all $i,j$. Note that $k \le \sum_{i\in I}|E(i)| = |E|$. We define a binary threshold function $g_i^l(E_i^j)=0$ if $g_i(E_i^j)\le R_l$, and $1$ otherwise. We iterate $l=1,2,\dots, k$ and run our general algorithm to minimize $\sum_{i\in I}g_i^l(M)$ until we find a stable matching with a total sum of $0$. This requires at most $|E|$ iterations.\end{remark}

\subsection{Additional Properties}

\begin{lemma}\label{lem:l1-lattice}
The optimal stable matchings with respect to $\sum_{i\in I}g_i(M)$ (or $\max_{i\in I}g_i(M)$) form a lattice, and thus admit a student-optimal solution.
\end{lemma}

\begin{proof}By Theorem~\ref{thm:main} and Remark~\ref{rem:maxcase}, these optimal stable matchings correspond to minimum weight stable matchings with respect to some weight function $\w$. Let $M$ and $M'$ be two such optimal matchings, meaning $\w(M)=\w(M')$. If we define $M\vee M'$ and $M\wedge M'$ as the matchings where students select their more preferred and less preferred institution between $M$ and $M'$ respectively, then it is well-known that these operations produce stable matchings (the join and meet) \citep{irving1987efficient}. Since each edge is counted the same amount of times in $M$,$M'$ and in $M\vee M'$, $M\wedge M'$ we have $\w(M)+\w(M')=\w(M\wedge M')+\w(M\vee M')$. Since $M$ and $M'$ are already minimal, $M\wedge M'$ and $M\vee M'$ cannot have strictly lesser weights. Therefore, we must have $\w(M\wedge M')=\w(M\vee M')=\w(M)$. This implies the set of optimal matchings is closed under meet and join operations, forming a lattice, which inherently guarantees the existence of a student-optimal solution.
\end{proof}
\begin{remark}If we are given an external cost function $z:E\to \mathbb{R}$, we can find a minimum $z$-cost \utilopt\ stable matching by defining a combined objective to make the value of $g_i$ lexicographically more important: $$g_i(M)=\frac{g_i(M)}{\min\limits_{\{i,j,j'\} :g_i(E_i^j)\ne g_i(E_i^{j'})}{|g_i(E_i^j)-g_i(E_i^{j'})|}}\cdot \max_{e\in E}|z(e)|\cdot (|E|+1) +\sum_{e\in M(i)}z(e).$$ 
\end{remark}

By Lemma~\ref{lem:l1-lattice}, if we assign small "tie-breaking" costs according to student preferences (where highly preferred edges have slightly smaller cost) as $z$, the computed minimum weight stable matching will naturally give us the student-optimal \utilopt\ or \egalopt\  stable matching (for the latter, using the fact that the \egalopt\ stable matchings are exactly the \diverse\ ones with the optimal binary threshold functions from Remark~\ref{rem:maxcase}).

\begin{theorem}
   A student optimal \utilopt\ stable matchings can be found in time $O(|S\cup I|^2(|S\cup I|\log^2|S\cup I|+ T_g))$ and an \egalopt\ stable matching can be found in time $O(|E||S\cup I|^2(|S\cup I|\log^2|S\cup I|+ T_g))$.
\end{theorem}

\subsection{Extension to Two-Sided Objectives}

A remarkable and immediate consequence of our approach is that it naturally extends to two-sided objective functions. %Because of Theorem~\ref{fact:stable_diff}, analogously to Lemma~\ref{lem:stabsetnum} we get that the set of stable sets is linearly small for all agents on both sides of the market, so we can simultaneously optimize complex set functions for both institutions and students. 

Suppose that, in addition to the institution objectives $g_i$, we are also given polynomial-time computable set functions $h_s:2^{E(s)}\to \mathbb{R}$ for each student $s\in S$. In a standard many-to-one school choice setting, $E(s)$ consists of at most one matched edge, making $h_s$ a simple edge-cost function; however, this extension holds even in many-to-many matching markets, where we also have a capacity $q_s$ for each $s\in S$. We want to find a stable matching $M$ that minimizes the global two-sided objective:$$\sum_{i\in I}g_i(M(i)) + \sum_{s\in S}h_s(M(s))$$

The symmetry of Theorem~\ref{thm:stable_diff} --analogously to Corollary~\ref{cor:stab-sets}-- guarantees that students also have a strictly limited number of nonempty stable sets $E_s^1, \dots, E_s^{k_s}$, $k_s\le E(s)$, which can be computed in polynomial time using the exact same iterative DA approach (with the roles of students and institutions reversed).

This allows us to state the following general theorem:

\begin{theorem}\label{thm:two-sided}
Given a many-to-many market $(S,I,E,q_I,q_S,\succ_I,\succ_S)$ with set functions $g_i:2^{E(i)}\to \mathbb{R}$ for all $i\in I$ and $h_s:2^{E(s)}\to \mathbb{R}$ for all $s\in S$, a stable matching $M$ that minimizes $\sum_{i\in I}g_i(M(i)) + \sum_{s\in S}h_s(M(s))$ can be found in $O(|S\cup I|^2(|S\cup I|\log^2|S\cup I|+ T_g+T_h))$ time, where $T_g$ and $T_h$ denote the maximum steps it takes to compute $g_i$ or $h_i$ for a stable set.
\end{theorem}

\begin{proof}We apply Step 1 and Step 2 of our general algorithm to the institutions to compute a linear edge weight function $\w_I(e)$ such that for any stable matching $M$, $\sum_{e \in M} \w_I(e) = \sum_{i\in I}g_i(M(i))$. By the exact same logic, we independently apply Step 1 and Step 2 to the students' side of the market. This yields a second linear edge weight function $\w_S(e)$ such that for any stable matching $M$, $\sum_{e \in M} \w_S(e) = \sum_{s\in S}h_s(M(s))$.

We then define a combined edge weight function as the sum of the two independent weights: $\w(e) = \w_I(e) + \w_S(e)$. Because any stable matching $M$ must restrict to exactly one stable set at each institution and one stable set at each student, the total weight of the matching perfectly matches the two-sided objective:
$$\w(M) = \sum_{e\in M} \w_I(e) + \sum_{e\in M} \w_S(e) = \sum_{i\in I}g_i(M(i)) + \sum_{s\in S}h_s(M(s))$$
Therefore, the problem reduces to computing a minimum $\w$-weight stable matching, which can be solved in polynomial time.
\end{proof}

\section{Simulations}
\label{sec:simulations}

% We complement the theoretical results with simulations on randomly generated
% school-choice markets with overlapping diversity categories and families of
% siblings. The simulations have two goals.  First, they compare the stable
% optimization algorithms to DA and to natural soft (resp. hard) quota-based diversity prioritizing DA algorithms.  Second,
% they quantify the tradeoff between diversity, family cohesion, student welfare, and
% standard stability.

We complement our theoretical results with simulations on randomly generated markets featuring overlapping diversity categories and sibling families. The simulations serve two purposes. First, we compare our stability-preserving optimization algorithms with the standard Deferred 
Acceptance algorithm and with diversity-prioritizing DA-based approaches that relax stability to better satisfy distributional objectives. Second, we 
quantify the trade-offs among diversity, family cohesion, student welfare, and classical stability by evaluating both distributional outcomes and the number 
of blocking pairs.

\paragraph{Instance generation.}
Each instance has \(n=500\) students and \(m=15\) institutions.  Total capacity
equals the number of students; in the reported runs this gives capacity
\(q_i\in \{ 33,34\} \) at every institution \(i\).  Student and institution locations are
drawn independently and uniformly from the unit square, and distances are
Euclidean.  Every student finds every institution acceptable, and every
institution finds every student acceptable.

Student preferences and institution priorities follow the spatial random model
of Erdil and Ergin~\citep{ErEr08a}.  For each institution \(i\), we draw a
common institutional utility \(Z_{0i}\sim N(0,1)\).  For each student-institution
pair \((s,i)\), we draw a student specific utility \(Z_{si}\sim N(0,1)\) for $i$.  Student
\(s\)'s utility from institution \(i\) is
\[
    U_{si}
    =
    -\beta d(s,i)
    +
    (1-\beta)\left(\alpha Z_{0i}+(1-\alpha)Z_{si}\right),
\]
where \(d(s,i)\) is the distance between the student and institution locations.
Preferences are obtained by ranking institutions in decreasing order of
\(U_{si}\), with random tie-breaking.  The parameter \(\beta\) controls the
importance of distance, and \(\alpha\) controls the amount of common agreement
across students about institutional quality.  In the main study we use
\[
    \alpha\in\{0.1,0.4\},
    \qquad
    \beta\in\{0.1,0.4,0.8\},
\]
with 50 independent instances per parameter pair.  

Institution priorities also follow the Erdil--Ergin walk-zone construction.  A
student's walk-zone institution is the closest institution to her location.
Each institution ranks all students in its own walk zone above all students
outside its walk zone.  Within each of these two priority classes, the
institution uses an independent random lottery.  Thus we use
school-specific multi-tie-breaking rules rather than a single common lottery.

Families are generated by first randomly ordering the students and then forming
families sequentially.  Each new family has size 1, 2, or 3 with probabilities
0.40, 0.50, and 0.10, respectively, subject to truncation at the end of the
student list.  Siblings share the same home location and the same demographic
profile.  In the reported simulations, sibling preferences are otherwise drawn
independently: siblings share distance terms and the common institutional utility $Z_{0i}$,
but have independently drawn $Z_{si}$ utilities, making them correlated, but not the same.

There are five categories in the implementation. Each student first receives exactly one socioeconomic
status category: low SES with probability \(0.50\), middle SES with probability
\(0.30\), and high SES with probability \(0.20\).  Conditional on SES, the
student is classified as an English-language learner with probability \(0.25\),
\(0.05\), or \(0.01\), respectively.  Independently, the student is classified
as special-education with probability \(0.12\).  Hence categories overlap:
students count for each category to which they belong.  Siblings receive the
same category set.

For each institution \(i\), category \(t\), and capacity \(q_i\), let \(N_c\)
be the total number of students in category \(c \), and define the proportional
target \(e_{ic}=N_c q_i/n\).  The lower and upper bounds are
\[
    \ell_{i}^c=\lfloor 0.9 e_{ic}\rfloor,
    \qquad
    u_{i}^c=\min\{q_i,\lceil 1.1 e_{ic}\lceil \}.
\]

% \paragraph{Metrics.}
% For a matching \(M\), we compute the total diversity violation $\divv_1(M)$ and the
% maximum diversity violation $\divv_\infty(M)$.

% The family score is the number of families whose members are all assigned to the
% same institution.  Average rank is the average, over students, of the rank of
% the assigned institution in the student's preference list; ranks are one-based,
% and an unmatched student with a list of length \(r\) receives rank \(r+1\).
% Finally, blocking edges are ordinary student-institution blocking pairs under
% the original priorities and capacities, without taking category quotas into
% account.  

\paragraph{Metrics.}

We evaluate each matching \(M\) using stability, distributional, and welfare-related metrics. 

First, we measure classical stability by the number of blocking edges. A blocking edge is defined with respect to the original student preferences, institutional priorities, and capacities, without incorporating any distributional constraints into the blocking condition.

Second, for distributional objectives, we compute the total diversity violation \(\divv_1(M)\) and the maximum diversity violation
\(\divv_\infty(M)\). Lower values indicate better satisfaction of diversity objectives.

Third, for experiments involving families, the family score measures the number of families whose members are all assigned to the same institution. Higher values indicate better family cohesion.

Finally, we measure student welfare by average rank, defined as the average rank of each student's assigned institution in their preference list. Ranks are one-based, and an unmatched student with a preference list of length \(r\) receives rank \(r+1\). Lower average rank indicates higher student welfare.

\paragraph{Algorithms.}

We compare six mechanisms, consisting of four stability-preserving algorithms and two DA-style distribution-prioritizing benchmarks.

The stability-preserving algorithms are DA, DC-DA, General-TotalViol, and General-Family. DA is the student-proposing deferred acceptance algorithm. DC-DA is the diversity-constrained DA algorithm described above, which directly minimizes diversity violations while preserving classical stability. General-TotalViol applies our general stable-optimization framework to the utilitarian diversity-violation objective, providing an implementation of the general reduction. General-Family applies the same framework to the objective of maximizing the number of families whose members are assigned to the same institution.

We additionally include two DA-style benchmark mechanisms that prioritize distributional objectives during the proposal process. These mechanisms do not generally preserve classical stability, but they provide natural comparisons by directly incorporating category constraints into institutional choice rules.

The \textbf{Smart-Hard-Reserves} mechanism of \citet{SoYe22a} is a DA-based reserve mechanism that prioritizes lower-bound diversity requirements. For each institution \(i\), category \(c\) receives \(\ell_i^c\) reserved seats, together with an open-reserve category containing the remaining seats. Students can access reserved seats only for categories to which they belong. The institutional choice rule selects the highest-priority feasible set of students satisfying these reserve constraints.

This mechanism is designed for one-to-one counting, where each student occupies a single reserved category. Under the one-to-all counting convention considered in our experiments, exclusive reserves may create additional restrictions because students contribute simultaneously to multiple categories while occupying only one seat. Moreover, the resulting choice function does not generally preserve classical stability.

The \textbf{Soft-Reserves} benchmark is a DA-based heuristic that directly incorporates diversity quotas during the proposal process. When an institution exceeds capacity, it first removes students causing upper-quota violations whenever possible, and then removes students whose rejection preserves lower-quota feasibility. We extend this rule from single-category settings to overlapping categories under the one-to-all counting convention. Similar to other quota-prioritizing DA mechanisms, Soft-Reserves may violate classical stability.

We include these two mechanisms as distribution-prioritizing benchmarks. In contrast, our algorithms optimize diversity objectives within the space of stable matchings. This comparison allows us to quantify the tradeoff between achieving stronger distributional objectives and preserving classical stability.

\paragraph{Findings.}
All figures report means with one-standard-deviation error bars.

\begin{figure}[hptb]
    \centering
    \begin{minipage}{0.48\textwidth}
        \centering
        \textbf{(a)} Total diversity violation\\
        \includegraphics[width=\textwidth]{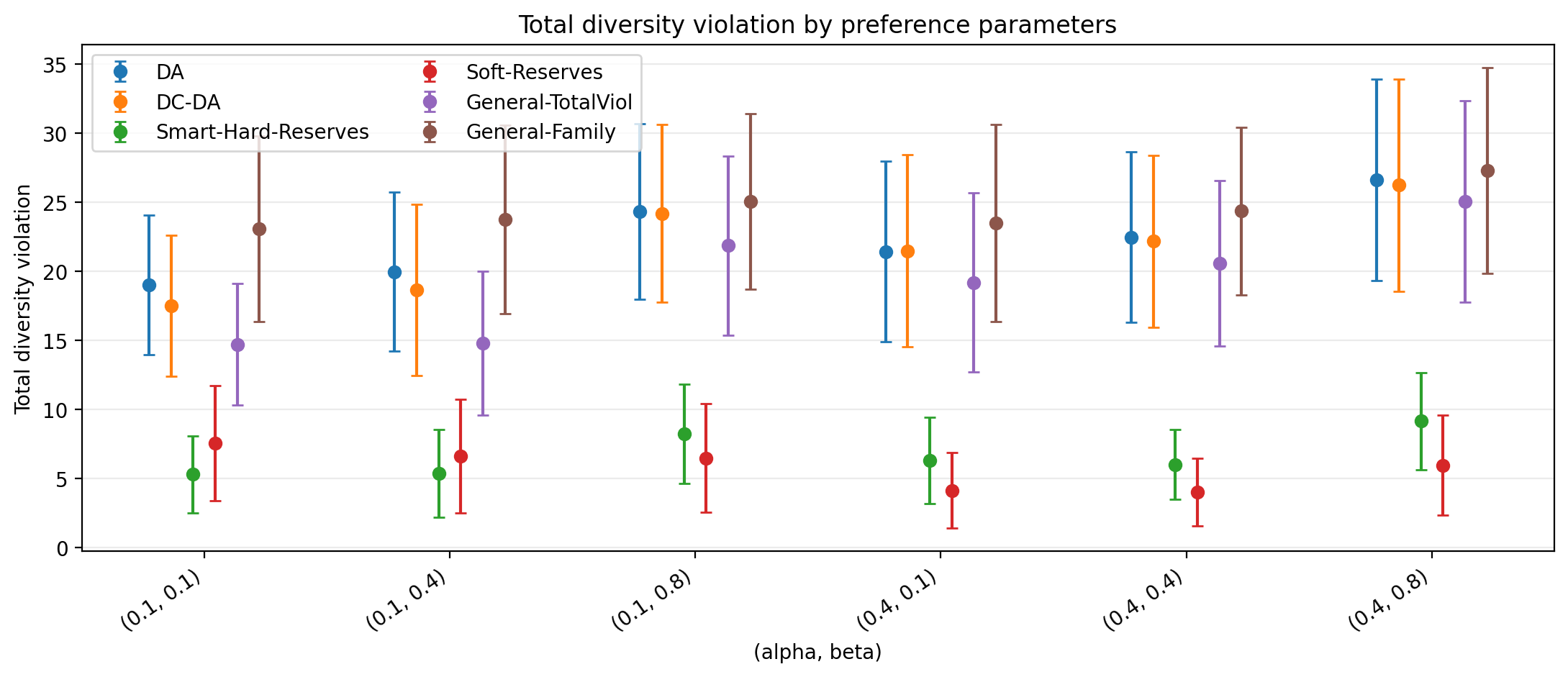}
    \end{minipage}
    \hfill
    \begin{minipage}{0.48\textwidth}
        \centering
        \textbf{(b)} Maximum diversity violation\\
        \includegraphics[width=\textwidth]{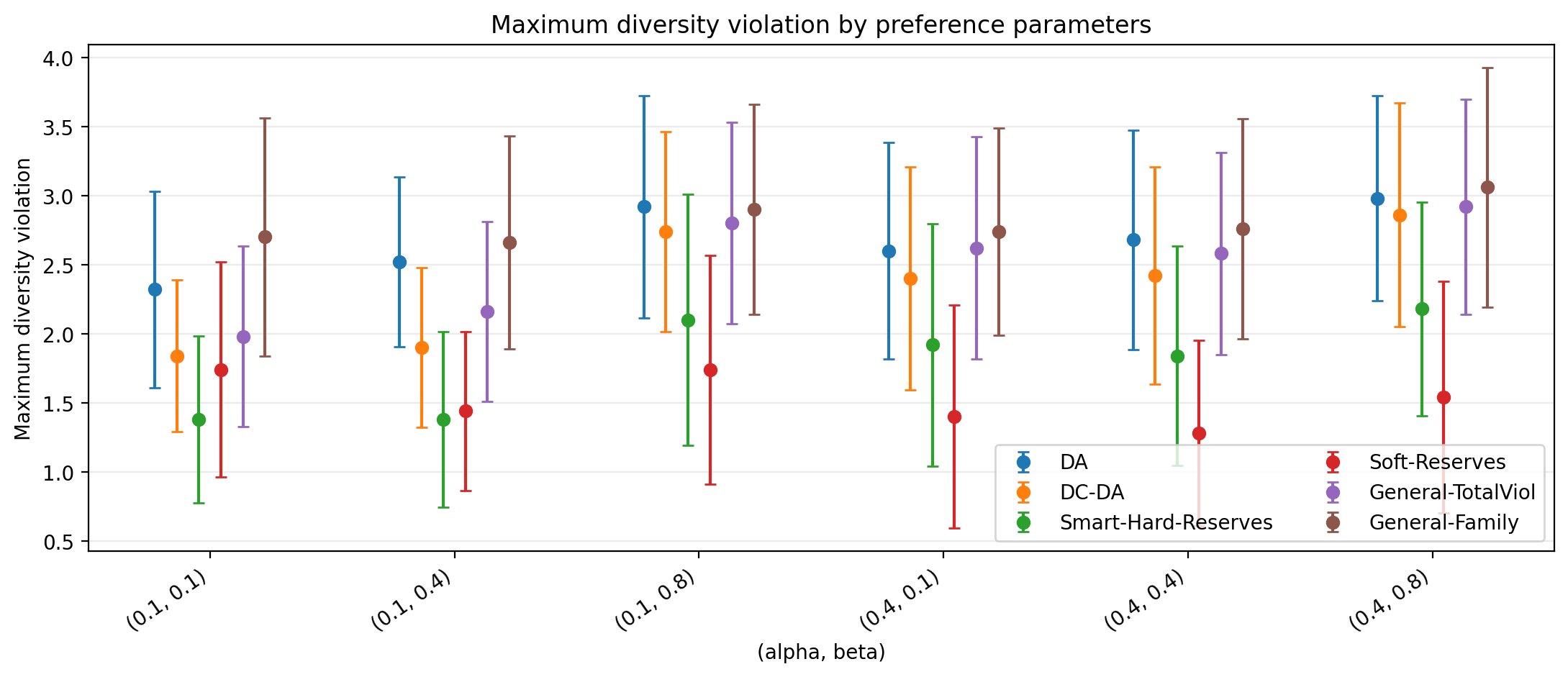}
    \end{minipage}

    \vspace{0.4cm}

    \begin{minipage}{0.48\textwidth}
        \centering
        \textbf{(c)} Families matched together\\
        \includegraphics[width=\textwidth]{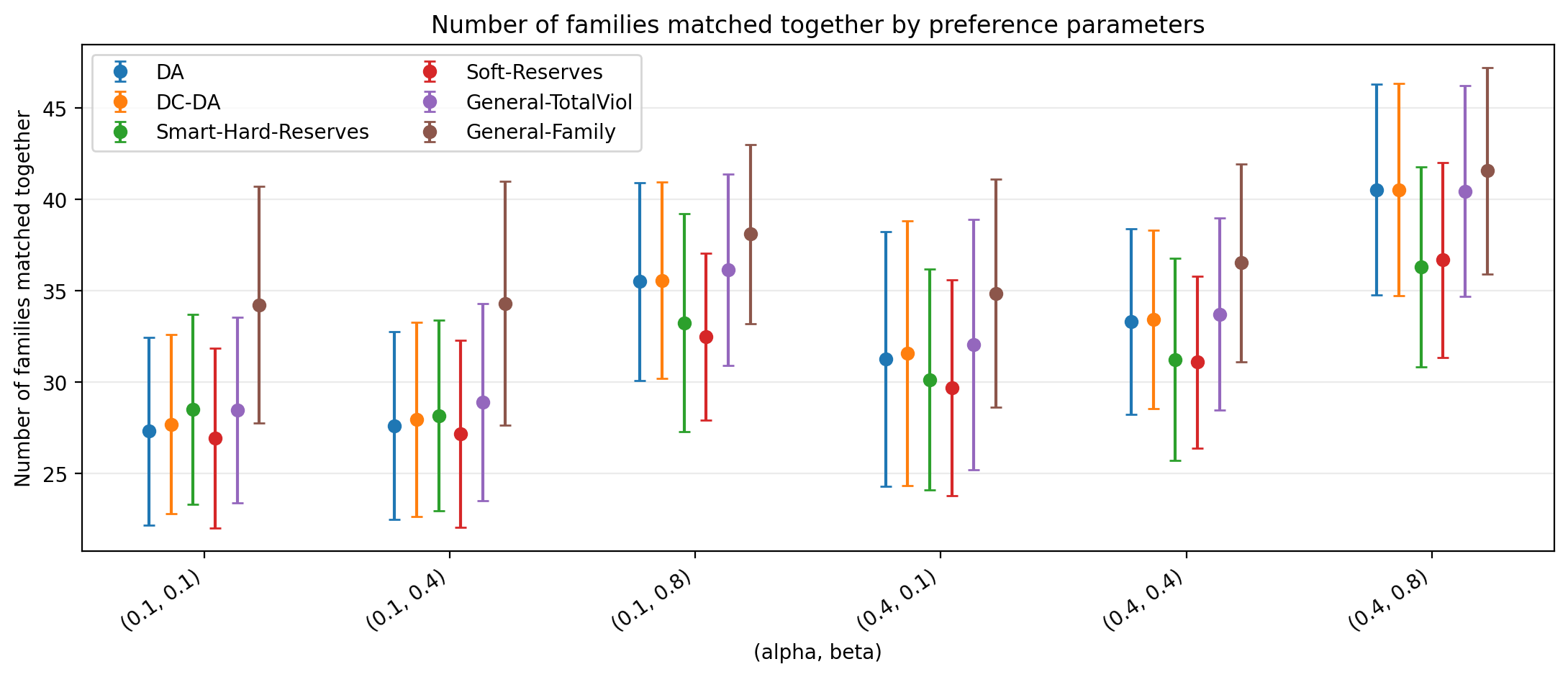}
    \end{minipage}
    \hfill
    \begin{minipage}{0.48\textwidth}
        \centering
        \textbf{(d)} Average assigned rank\\
        \includegraphics[width=\textwidth]{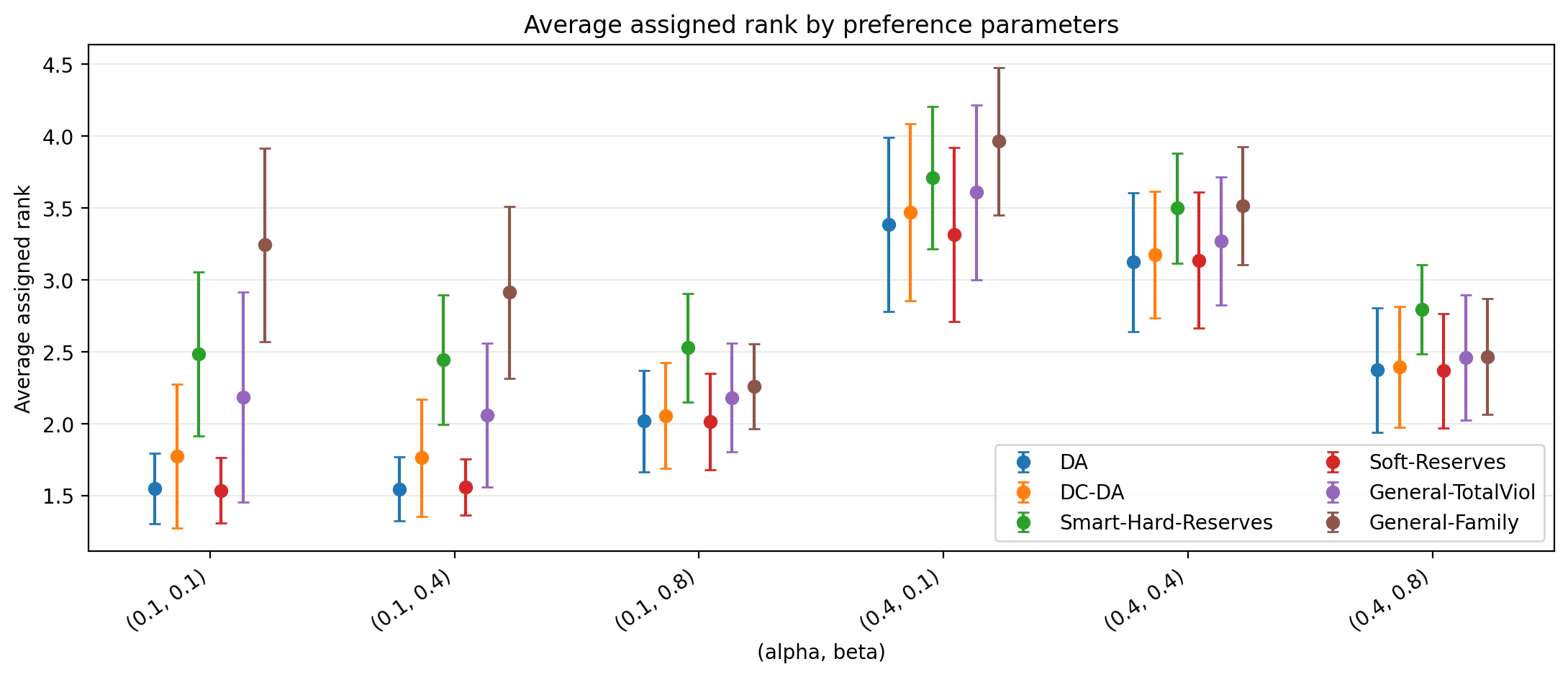}
    \end{minipage}

    \vspace{0.4cm}

    \begin{minipage}{0.62\textwidth}
        \centering
        \textbf{(e)} Blocking edges\\
        \includegraphics[width=\textwidth]{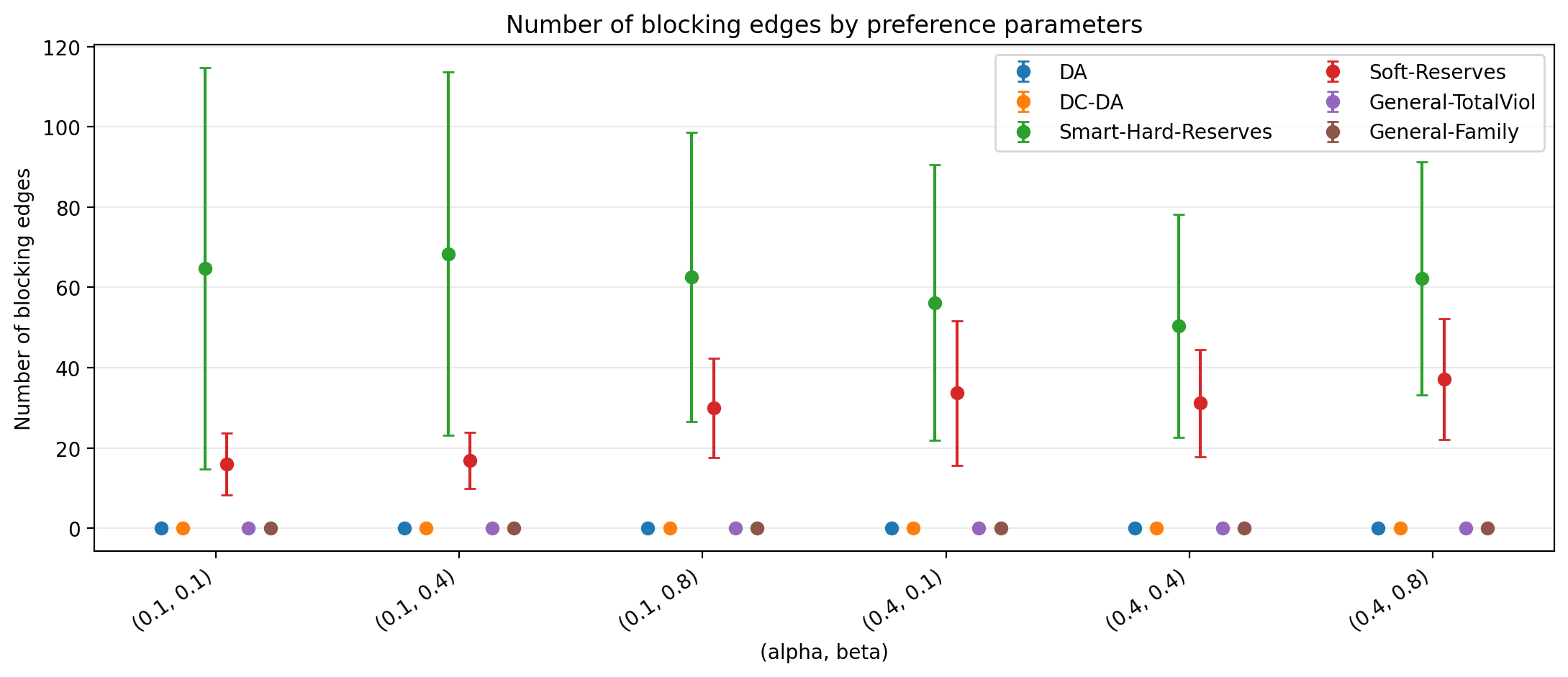}
    \end{minipage}
    \caption{Main alpha-beta simulation outcomes.  Points are averages over 50
    independent instances per parameter pair; vertical bars are one standard
    deviation.}
    \label{fig:alpha-beta-main}
\end{figure}

\begin{figure}[hptb]
    \centering
    \begin{minipage}{0.32\textwidth}
        \centering
        \includegraphics[width=\textwidth]{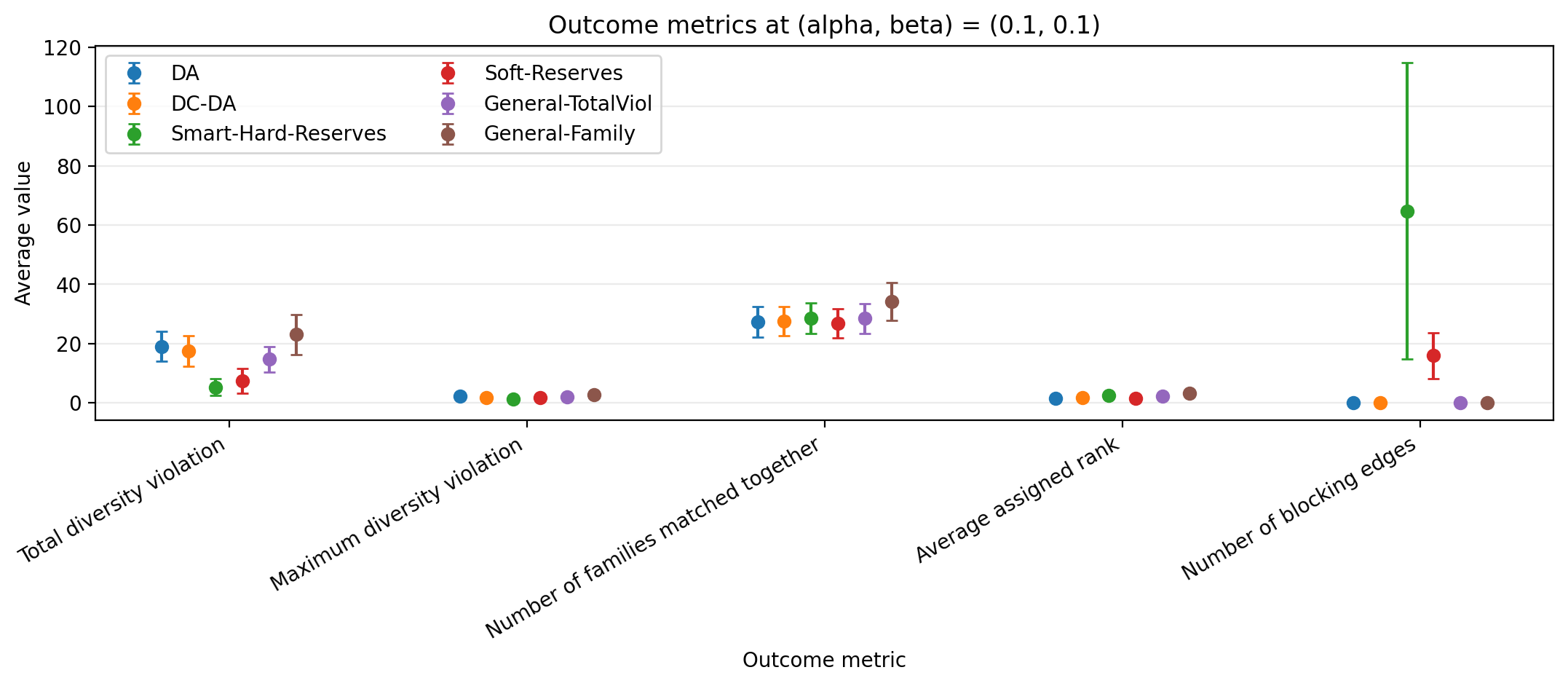}
    \end{minipage}
    \begin{minipage}{0.32\textwidth}
        \centering
        \includegraphics[width=\textwidth]{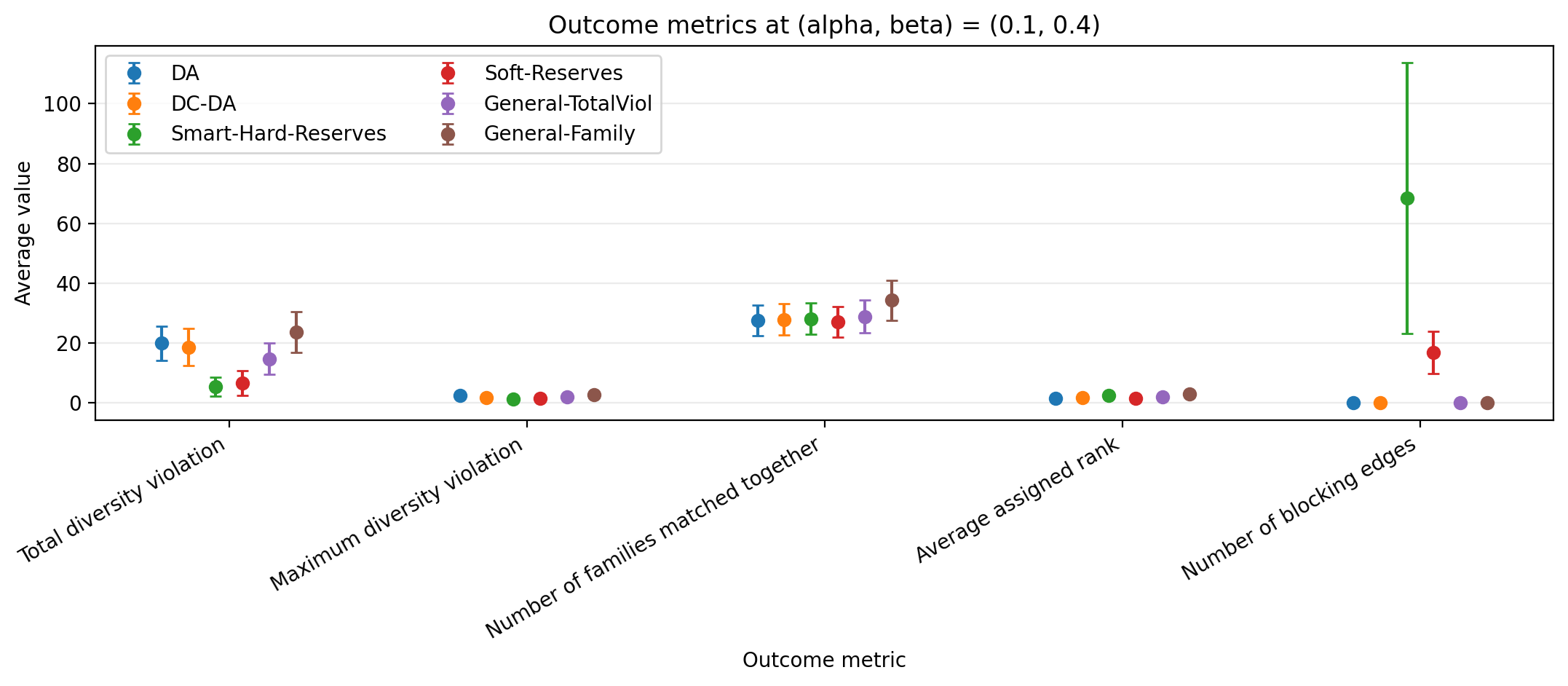}
    \end{minipage}
    \begin{minipage}{0.32\textwidth}
        \centering
        \includegraphics[width=\textwidth]{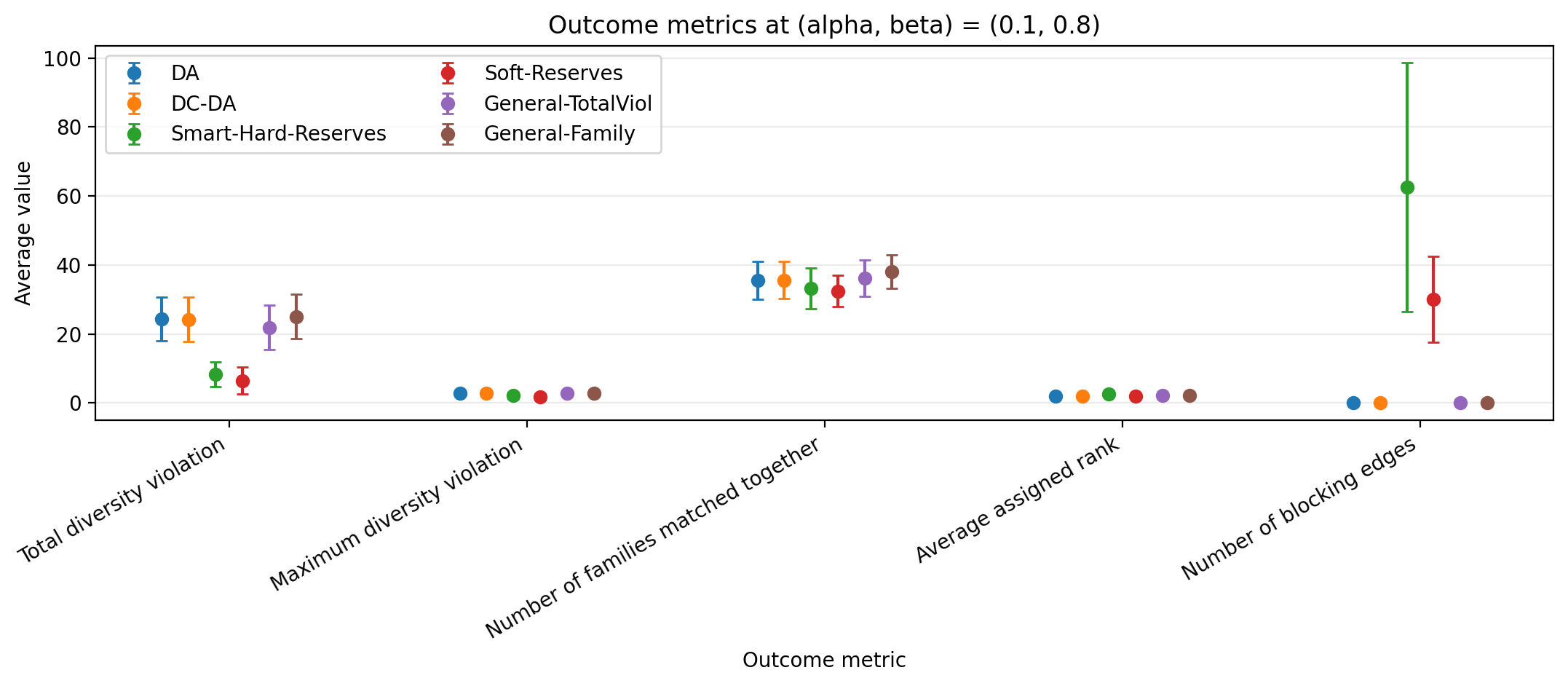}
    \end{minipage}

    \vspace{0.3cm}

    \begin{minipage}{0.32\textwidth}
        \centering
        \includegraphics[width=\textwidth]{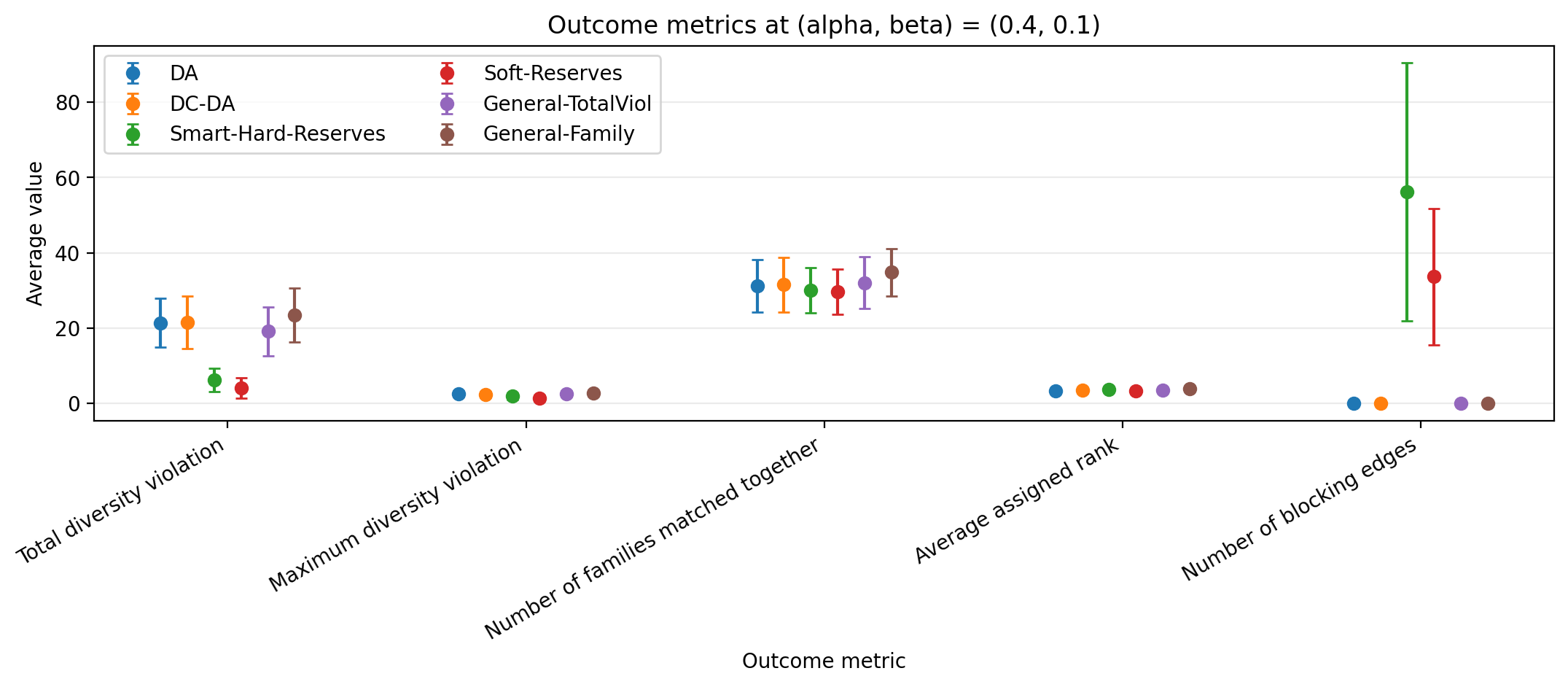}
    \end{minipage}
    \begin{minipage}{0.32\textwidth}
        \centering
        \includegraphics[width=\textwidth]{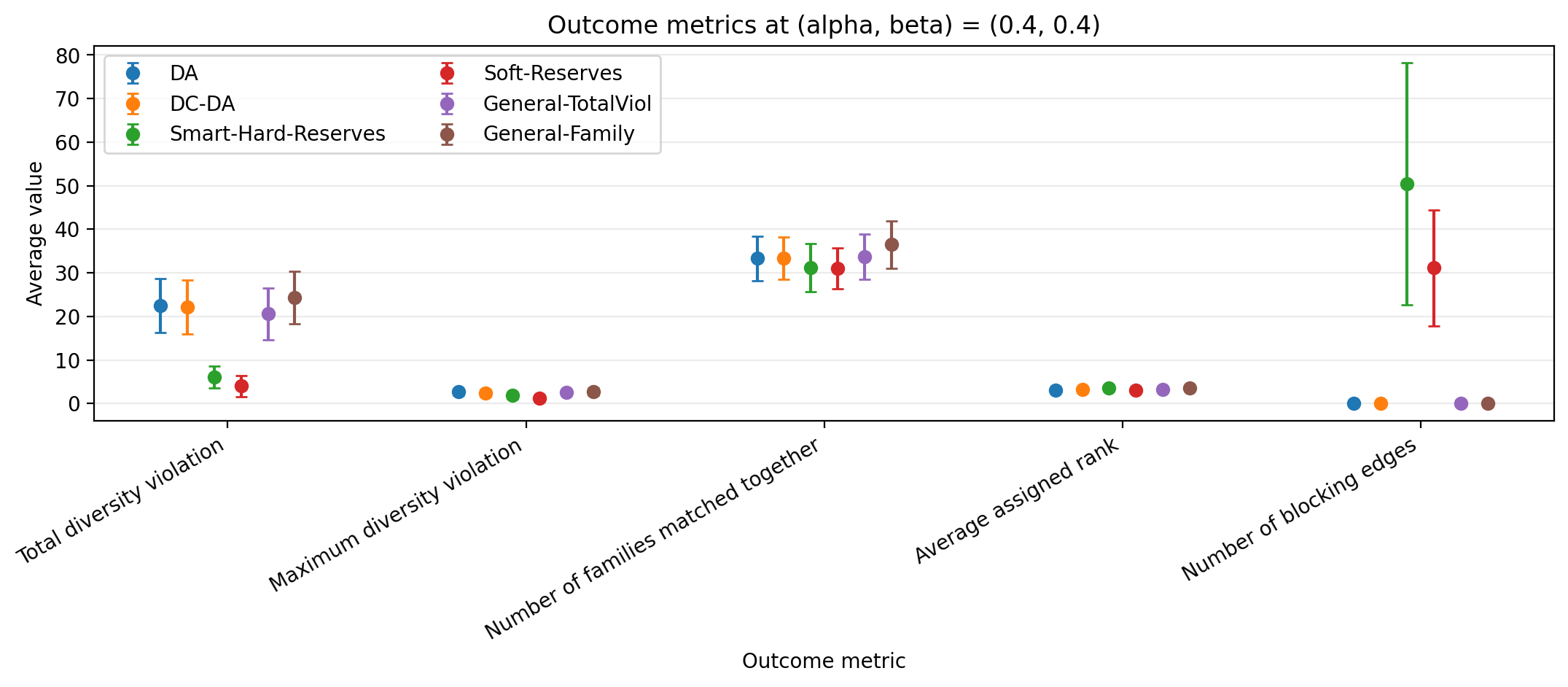}
    \end{minipage}
    \begin{minipage}{0.32\textwidth}
        \centering
        \includegraphics[width=\textwidth]{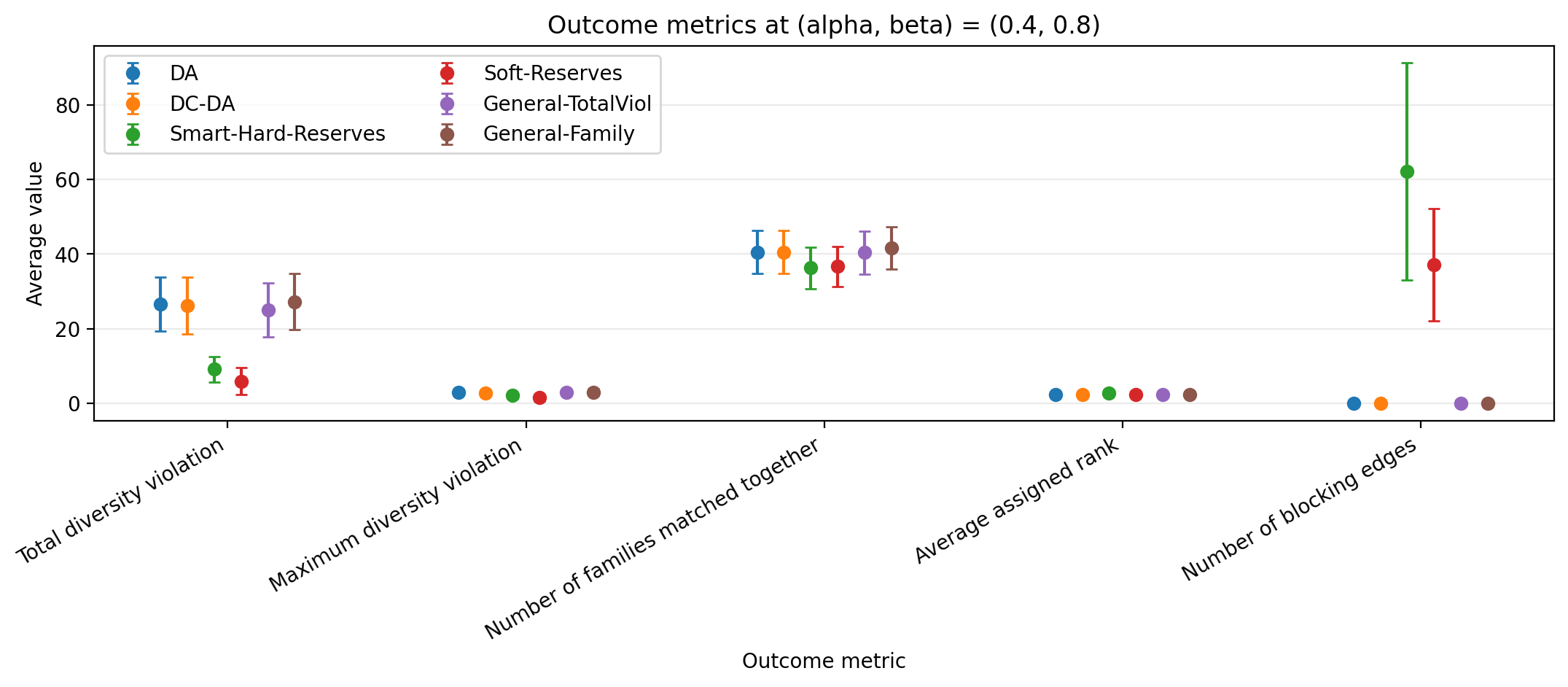}
    \end{minipage}
    \caption{Metric-by-algorithm view for each reported preference-parameter
    cell.  The x-axis lists the outcome metrics and the vertical bars are one
    standard deviation.}
    \label{fig:alpha-beta-metric-grid}
\end{figure}

\section{Conclusions}

In this paper, we resolved a fundamental tension between maintaining strict stability and satisfying complex distributional goals in centralized clearinghouses. Rather than weakening the definition of stability to accommodate rigid quotas, we demonstrated that distributional constraints can be effectively optimized as soft objectives directly over the lattice of standard stable matchings.

Our main contribution is a general, polynomial-time algorithmic framework that optimizes arbitrary, polynomial-time computable institutional and two-sided objectives. By leveraging the limited number of distinct stable sets an institution can face, we mapped non-linear set functions into linear edge weights, elegantly reducing the problem to finding a minimum-weight stable matching. This approach efficiently solves major practical challenges, such as minimizing overlapping diversity quota violations and maximizing sibling co-assignments. Coupled with our specialized Diversity-Constrained Deferred Acceptance (DC-DA) algorithm, these results equip market designers with flexible tools to achieve modern distributional goals without sacrificing the core fairness guarantee of stability.

% \clearpage
% \section{Experiments}

% \subsection{instance generator class}

% \begin{itemize}
%     \item input: number\_S, number\_C, number\_T, diversity function $g$?
%     \item functions:
%     \begin{itemize}
%         \item assign\_type\_to\_students:  Do we assume that agents’ types are independently distributed, with additional parameters such as the proportion of each type in the population?
%         \item set\_diversity\_goals:  how to defined the diversity function $g$?
%     \end{itemize}
%     \item output: datasets of instance
% \end{itemize}

% \subsection{Algorithms}

% possible algorithms
% \begin{itemize}
%     \item student-optimal DA
%     \item general algorithm 
%     \item \citep{SoYe22a}: smart reserve (maximal diversity) only works for maximum quota
%     \item \citep{AzSu25a}: rank-maximal smart reserve
%     \item 
% \end{itemize}

% same pipeline
% \begin{itemize}
%     \item input: instance
%     \item output: matching
% \end{itemize}

% \subsection{outcome class}
% \begin{itemize}
%     \item attributes:
%     \begin{itemize}
%         \item instance\_id
%         \item outcome\_algorithm\_dictionary: algorithm\_name, matching
%     \end{itemize}
%     \item functions:
%     \begin{itemize}
%         \item blocking\_number: blocking pair / envy / waste 
%         \item number of matches 
%         \item diversity\_goals: violation of function $g$?
%     \end{itemize}
% \end{itemize}

\bibliographystyle{plainnat}
\bibliography{references}  

\end{document}